\documentclass[10pt]{article}

\usepackage{natbib}
\usepackage{ams}
\usepackage{graphicx}
\usepackage{epsfig}
\usepackage{soul,color}
\usepackage{multirow}
\usepackage{booktabs}
\usepackage{amsmath}
\usepackage{amsthm}
\usepackage{setspace}
\usepackage{hyperref}

\newcommand{\tr}{^{\prime}}

\def\b#1{\mbox{\boldmath $#1$}}    
\def\bl#1{\mbox{\footnotesize \boldmath {$#1$}}} 

\renewcommand{\th}{\theta}

\newtheorem{prop}{Proposition}

\usepackage[nolists]{endfloat}

\usepackage{morefloats}

\begin{document}
%
\title{Three-step estimation of latent Markov models with covariates}
\author{Francesco Bartolucci\footnote{Department of Economics, University of Perugia, Via A. Pascoli, 20, 06123
Perugia.}  \and Giorgio E. Montanari\footnote{Department of Political Sciences, University of Perugia, Via A. Pascoli, 20, 06123
Perugia.} \and Silvia
Pandolfi$^*$\footnote{E-mail: pandolfi@stat.unipg.it}}
\date{}
\maketitle \vspace*{-0.5cm}


\begin{abstract}
\begin{singlespace}
We propose a modified version of the three-step estimation method for the latent
class model with covariates, which may be used to estimate latent Markov models for longitudinal data. 
The three-step estimation approach
we propose is based on a preliminary clustering 
of sample units on the basis of the time specific responses only.  This approach represents
an useful estimation tool when a large number of
response variables are observed at each time occasion. In such a context, full maximum likelihood estimation, which is typically based on the 
Expectation-Maximization algorithm, may have some drawbacks, essentially due to the presence of many local maxima of the 
model likelihood. Moreover, the EM algorithm may be particularly slow to converge, and may become unstable
with complex LM models. We prove the consistency of the proposed three-step estimator when the number of response variables tends to infinity. 
 We also show the results of a simulation study aimed at evaluating the performance of the
proposed alternative approach with respect to the full likelihood method. We finally illustrate an application to
a real dataset on the health status of elderly people hosted in Italian nursing homes. 
\end{singlespace}
\noindent \vskip7mm \noindent {\sc Keywords:} EM algorithm,
Latent class model, Nursing homes, Pseudo likelihood methods
\end{abstract}\newpage

%

\section{Introduction}  
In many applications involving longitudinal data, the interest is
often focused on the evolution of a latent characteristic of a group
of individuals over time, which is measured by one or more
occasion-specific response variables.  
This characteristic may
correspond, for instance, to the quality-of-life of subjects
suffering from a certain disease, to a certain type of ability, or to the tendency to commit crimes. 

The latent Markov (LM) models of \cite{wigg:73} may be usefully applied in such a context; for a review see \cite{bart:et:al:13}.
The main
assumption underlying these models, known as {\em local independence}, is that the response variables are
conditionally independent given a latent process which follows a Markov chain with a finite
number of states and which is typically homogeneous and of first order. The basic idea behind this assumption is that the latent process fully
explains the observable behavior of a subject; moreover, the latent
state to which a subject belongs at a certain occasion only depends
on the latent state at the previous occasion. 
In the presence of multivariate outcomes (more response variables for each time occasion),
the model allows for a dynamic clustering of subjects also depending on individual covariates.

LM models may be estimated by a Full Maximum Likelihood (FML) method based on the Expectation-Maximization (EM)
algorithm \citep{baum:et:al:70,demp:lair:rubi:77} and recursions which are well known in the 
hidden Markov literature \citep{baum:et:al:70,zucc:macd:09}.  
However, especially in the presence of many response variables for each time occasion and individual covariates, 
the FML estimation may have some drawbacks, as illustrated in more detail in the following.

In this paper, we propose an estimation approach 
which may be seen as an extension of the three-step estimation method for the
latent class (LC) model with covariates \citep{vermunt:2010}. 
This new estimation method is based on a preliminary dynamic clustering of subjects on the basis of only the time-specific responses, and
may be seen as a stable alternative to the  
full likelihood method, which also allows for fast estimation of the model parameters.

We recall that, when applied to the LC model with covariates, the three-step estimation approach consists of first fitting a basic LC
model for the set of response variables ignoring the covariates. Then, the subjects are assigned
to the latent classes on the basis of their posterior class membership
probabilities. Finally, the association between class membership and
individual covariates is investigated through the estimation of a
suitable multinomial logit model \citep[see, among others,][]{bolck:et:al:04,lu:thomas:08,bakk:et:al:13}. 

In the proposed version of the three-step approach, hereafter 3S, every sample unit is not
strictly assigned to a latent state at each occasion and this latent state can change 
across time. In this way, even the parameters affecting the transition probabilities of the 
latent Markov chain may be estimated. We apply this estimation method to the basic LM model
and to its extended version which includes individual covariates.

The performances of the proposed 3S estimation method are assessed through a Monte Carlo simulation study
in comparison with those of the FML method. We 
also illustrate the approach by an application based on a real dataset 
which derives from a project, named ULISSE,
on the health status of elderly people hosted in Italian nursing homes.

The paper is organized as follows. In Section
\ref{sec:lm} we outline the basic LM model and its extended version with individual covariates.
Section \ref{sec:EM} illustrates the estimation of the
model on the basis of the standard FML method, whereas 
in Section \ref{sec:three-step} we outline the proposed 3S estimation approach. In Section \ref{sec:5}
we illustrate the results of the simulation study, aimed at evaluating the 
performances of the proposed estimation approach, and in Section \ref{sec:ulisse} we show the results of the application to real data. 
Section \ref{sec:6} contains main conclusions.

\section{Latent Markov model}
\label{sec:lm}  
In the following, we illustrate the basic latent Markov (LM) model and its
extended version which includes individual covariates.

\subsection{Basic LM for multivariate data}
\label{sec:basic}
Let consider the multivariate case in which we observe a vector
$\b Y^{(t)}$ of $r$ categorical response variables, $Y_{j}^{(t)}$, with $c_j$ categories, from 0 to $c_j-1$, $j=1,\ldots,r$, which are available at the $t$-th
time occasion, $t=1,\ldots,T$.
Let also $\b Y$ be the vector made up of the union of the vectors  $\b Y^{(t)}$, $t=1,\ldots,T$.

The model assumes the existence of a
latent process $\b U =(U^{(1)},\ldots,U^{(T)})$
which affects the distribution of the response variables. Such a process is assumed
to follow a first-order Markov chain with state space
$\{1,\ldots,k\}$, where $k$ is the number of latent states. Under the
{\em local independence} assumption, the response vectors $\b Y^{(1)},\ldots,\b Y^{(T)}$ are
assumed to be conditional independent given the latent process $\b U$.  Moreover, the elements $Y_{j}^{(t)}$ within $\b Y^{(t)}$,
$t=1,\ldots,T$, are conditionally independent given $U^{(t)}$. This
assumption leads to a strong simplification of the model but it can be
relaxed by allowing conditional serial dependence or
conditional contemporary dependence \citep[for details, see][]{bart:et:al:13}. 
Note also that the LM model may be seen as an extension of the latent class (LC) model \citep{laza:50,laza:henr:68,good:74b}, 
in which the assumption that each subject belongs to the same latent class throughout  the survey is suitable 
relaxed.

The model formulation is based on the following conditional response probabilities:
\begin{equation}\label{eq:cond}
\phi_{jy|u}=p(Y^{(t)}_{j}=y|U^{(t)}=u), \quad j=1,\ldots,r,\;\;\; u=1,\ldots,k,\;\;\; y=0,\ldots,c_j-1,
\end{equation}
which are collected in the matrix $\b \Phi_j$ of dimension $c_j \times k$.
Moreover, parameters of the latent process are the initial probabilities
\[
\pi_{u}= p(U^{(1)}=u),\quad
u=1,\ldots,k,
\]
collected in the vector $\b \pi$, and the transition probabilities
\[
\pi_{v|u}= p(U^{(t)}=v|U^{(t-1)}=u), \quad t=2,\ldots,T,\;\;\; u,v=1,\ldots,k,
\]
collected in the matrix $\b \Pi$. Note that all these probabilities do not depend on the specific sample unit since,
in its basic version, the LM model does not include individual covariates. Moreover, in this version the conditional
response probabilities and the transition probabilities are time-homogeneous.

The above assumptions imply that the distribution of $\b U$ may be expressed as
\[
p(\b U = \b u) = \pi_{u^{(1)}}\prod_{t=2}^T\pi_{u^{(t)}|u^{(t-1)}},
\]
where $\b u = (u^{(1)},\ldots,u^{(T)})$ denotes a realization of $\b U$. Moreover, the conditional distribution of $\b Y$ given  
$\b U$ is defined as 
\[
p(\b Y= \b y|\b U= \b u) = \prod_{t=1}^T \phi_{\bl y^{(t)}|u^{(t)}},\]
where $\b y $ is made by the subvectors $\b y^{(t)}= (y_1^{(t)},\ldots,y_r^{(t)})$, and, in general, due to the assumption of local independence, 
we define 
\[
\phi_{\bl y^{(t)}|u} =
\prod_{j=1}^r \phi_{jy_j^{(t)}|u}.  
\]
Finally for the {\em manifest distribution} of $\b Y$ 
we have
\begin{eqnarray}
p(\b y) = p(\b Y=\b y) =  \sum_{\bl u}\pi_{u^{(1)}}\pi_{u^{(2)}|u^{(1)}}\cdots
\pi_{u^{(T)}|u^{(T-1)}}\phi_{\bl y^{(1)}|u^{(1)}}\phi_{\bl y^{(2)}|u^{(2)}}\cdots
\phi_{\bl y^{(T)}|u^{(T)}}. \label{eq:dist3}
\end{eqnarray}
It is important to note that computing $p(\b y)$ as expressed in (\ref{eq:dist3})
involves a sum extended to all the possible
$k^T$ configurations of the vector $\b u$;  this typically requires
a considerable computational effort.
In order to efficiently compute such a probability, we can use the Baum-Welch
forward recursion \citep{baum:et:al:70,welch:2003}; see \cite{bart:et:al:13} for details.

\subsection{Extended LM model with individual covariates}\label{sec:lmcov} 

Let $\b x^{(t)}$ denote the vector of individual
covariates which are available at the $t$-th time occasion, and let $\b x$ denote the vector of all the
individual covariates, which is obtained by stacking the vectors $\b
x^{(1)},\ldots,\b x^{(T)}$. In this paper, we consider the case in which the covariates are included
in the latent model, so as to affect the initial and transition probabilities. In this situation,
the assumption of local independence and the assumption that the latent process is 
of first order still hold.  Moreover, the conditional response
probabilities are again time-homogeneous and equal for all subjects as in (\ref{eq:cond}).
The initial and transition probabilities depend instead on the individual covariates
through a multinomial logit model
\begin{eqnarray}  \label{eq:be}
\log \frac{p(U^{(1)}=u|\b x^{(1)})}{p(U^{(1)}=1|\b x^{(1)})} &=& 
\beta_{0u}+(\b x^{(1)})\tr \b \beta_{1u},\quad u\geq 2,\\ \label{eq:Ga}
\log \frac{p(U^{(t)}=v|U^{(t-1)}=u,\b x^{(t)})}
{p(U^{(t)}=u|U^{(t-1)}=u,\b x^{(t)})}&=&\gamma_{0uv}+(\b x^{(t)})\tr
\b\gamma_{1uv}, \quad t\geq 2,\;\;\;u\neq v. 
\end{eqnarray}
In the above expressions, $\b\beta_u=(\beta_{0u},\b\beta_{1u}\tr)\tr$ and $\b\gamma_{uv}=(\gamma_{0uv},\b\gamma_{1uv}\tr)\tr$ are parameter vectors to be estimated which are collected in the matrices $\b\beta$ and  $\b\Gamma$.

The manifest distribution of $\b Y$ 
is now conditional to $\b x$ and may be easily
defined as an extension of (\ref{eq:dist3}). The corresponding probability function is denoted by 
$p(\b y | \b x)$.

\section{Maximum likelihood estimation}
 \label{sec:EM}

Given a sample of $n$ independent units that provides the response vectors 
$\b y_1, \ldots, \b y_n$, and given the vectors of covariates $\b x_1,\ldots,\b x_n$, the model log-likelihood assumes the following expression
\[
\ell(\b\th) = \sum_{i=1}^n\log p(\b y_i|\b x_i ).
\]
Note that in this case each vector $\b y_i$ is a realization of $\b Y_i$ which is made up of the 
subvectors $\b y_i^{(t)}$, $t=1,\ldots,T$, that, in turn, have elements $y_{ij}^{(t)}$, $j=1,\ldots,r$. 
Moreover, $\b x_i$ may be decomposed into the time-specific subvectors of covariates
$\b x_i^{(1)},\ldots,\b x_i^{(T)}$. Note also that  
$p(\b y_i|\b x_i )$ corresponds to the
manifest probability of the responses provided by subject $i$, given the covariates,  
and $\b\th$ is the vector of all free parameters affecting $p(\b y_i| \b x_i)$.
The likelihood function can be maximized by the EM algorithm \citep{baum:et:al:70,demp:lair:rubi:77}, as described in the following section.
\subsection{Expectation-Maximization algorithm}
The EM algorithm is based on the complete data log-likelihood that we could compute
if we knew the latent state of every subject at each occasion. The complete data
are represented by the pairs $(\b u_i,\b y_i)$, $i=1,\ldots,n$, where $\b u_i =(u_i^{(1)},\ldots,u_i^{(T)} )$ is the 
sequence of latent states of subject $i$.  The complete data log-likelihood has the following expression 
\begin{eqnarray}
\hspace*{-1.5cm}\ell^*(\b\th) &=&\sum_{i=1}^n \Bigg[ \sum_{j=1}^r \sum_{t=1}^T\sum_{u=1}^k\sum_{y=0}^{c_j-1} a_{ijuy}^{(t)}\log \phi_{jy|u}+\sum_{u=1}^kb_{iu}^{(1)}\log p(U_i^{(1)}=u|\b x_i^{(1)})
\nonumber\\
&& +
\sum_{t=2}^T\sum_{u=1}^k\sum_{v=1}^k b_{i u v}^{(t)}\log p(U_i^{(t)}=v|U_i^{(t-1)}=u,\b x_i^{(t)})\bigg],\label{eq:comp_lk_cov}   
\end{eqnarray}
where $a_{ijuy}^{(t)}$ is the indicator variable for subject $i$ responding by $y$ 
at occasion $t$ to response variable $j$ and belonging to latent state $u$ at the same occasion,
$b_{iu}^{(t)}$ is the indicator variable for subject $i$ being in latent state $u$ at occasion $t$, 
whereas  $b_{iuv}^{(t)}$ is the indicator variable for the transition, of subject $i$, from state $u$ at occasion $t-1$ to state
$v$ at occasion $t$. 

The EM algorithm alternates the following two steps until convergence:
\begin{itemize}
\item {\bf E-step}: it consists in computing the posterior expected value of each indicator variable 
involved in (\ref{eq:comp_lk_cov}) by suitable forward-backward recursions \citep{baum:et:al:70};
\item {\bf M-step}: it consists in maximizing the complete data log-likelihood
expressed as in (\ref{eq:comp_lk_cov}), with each indicator variable
substituted by the corresponding expected value. How to maximize
this function depends on the specific formulation of the model and,
in particular, on whether the covariates are included in the
measurement model or in the latent model.
\end{itemize}

Even if the EM algorithm is typically used to estimate LM models, in real 
applications the FML estimators may have some drawback. In particular, it is not uncommon 
that the number of response variables for each time occasion is very large, as well as the number of estimated latent
states. As an example, in the application illustrated in Section \ref{sec:ulisse},
the dataset consists of 
$r=75$ polytomous items, with a different number of response categories, and 28 individual covariates, for which we estimate
an LM model with $k=4$ latent states. 
In similar circumstances, the model likelihood may present many local maxima, requiring then a very high number
of random starting values to find the global maximum.
Moreover, the EM algorithm may be particularly slow to converge, and may become unstable with
complex LM models, which also include a large number of time-fixed or time-varying covariates. All these problems 
cannot be solved by a direct maximization algorithm, such as the Newton-Raphson algorithm 
which is known to have
instability problems even in simpler cases \citep[see, among others,][]{turner:08}. 
The proposed 3S estimation approach allows us to overcome some of these 
problems without being computationally intensive. Moreover, this strategy may also be useful to 
define suitable starting values for the EM algorithm to be used in the full-maximum likelihood approach.

\section{Proposed three-step estimation approach}\label{sec:three-step}

In the literature on LC models, the three-step approach \citep{vermunt:2010} represents an 
alternative to the FML method to deal with the inclusion of individual covariates.  The latter approach 
may indeed be impractical in some applications, especially when the number of covariates is large, as typically happens
in exploratory studies.  Moreover, the simultaneous modeling of covariates and latent
variables may introduce additional complexity in selection and estimation of models, with a potentially large number 
of parameters \citep{bakk:et:al:13}.  

In this section, we introduce a modified version of the three-step 
approach in order to estimate the basic 
LM model and its extended version with individual covariates. More in detail,
the proposed 3S method is based on the following steps:
\begin{itemize}
\item {\bf Step 1:} Fit a basic LC model for the set of response
variables in which the responses provided by the same
sample unit at different occasions are considered as corresponding to
separated units. On the basis of this preliminary fitting, performed via the EM algorithm, we
obtain the final estimates of the conditional response probabilities,
$\hat{\phi}_{jy|u}$, and the ``temporary'' estimates of the marginal probabilities of the latent states
 $\hat{\rho}_u=\hat{p}(U_i^{(t)} = u)$.
\item {\bf Step 2:} For each subject $i$ compute the posterior expected value of $b_{iu}^{(t)}$ and $b_{iuv}^{(t)}$ (see equation (\ref{eq:comp_lk_cov})) only on the basis of the results of the first step as \vspace{1mm}
\begin{eqnarray*}
\tilde{b}_{iu}^{(t)}&\propto&\hat{\rho}_u\:\hat{p}( \b Y_i^{(t)} = \b y_i^{(t)}|U_i^{(t)} = u)=
\hat{\rho}_u\prod_j\hat{\phi}_{jy_{ij}^{(t)}|u}, \quad t=1,\ldots,T,\\
\tilde{b}_{iuv}^{(t)}&=&\tilde{b}_{iu}^{(t-1)}\tilde{b}_{iv}^{(t)},\quad t = 2,\ldots,T.
\end{eqnarray*}  
\item {\bf Step 3:} Maximize the components of the complete data log-likelihood 
involving the latent structure parameters
\begin{eqnarray}\label{eq:lbe}
\tilde{\ell}^*_1(\b\beta) &=&\sum_i \sum_u \tilde{b}_{iu}^{(1)}\log p(U_i^{(1)} = u|\b x_i^{(1)}),\\  \label{eq:lGa}
\tilde{\ell}^*_2(\b\Gamma) &=&\sum_i\sum_{t\geq 2}\sum_u \sum_v \tilde{b}_{iuv}^{(t)}\log p(U_i^{(t)}=v|U_i^{(t-1)}=u,\b x_i^{(t)}).
\end{eqnarray}
\end{itemize}

For the basic LM model, the third step simply consists in computing
\begin{eqnarray*}
\tilde{\pi}_u &\propto& \sum_i\tilde{b}_{iu}^{(1)},\\ 
\tilde{\pi}_{v|u} &\propto& \sum_i\sum_{t\geq 2}\tilde{b}_{iuv}^{(t)}.
\end{eqnarray*} 
When we deal with the extended LM model with covariates, the last step consists of estimating $\b \beta$
and $\b \Gamma$, defined in (\ref{eq:be}) and (\ref{eq:Ga}), by fitting multinomial logit models based on a weighed likelihood as defined in (\ref{eq:lbe}) and (\ref{eq:lGa}); the corresponding estimates are denoted
by $\tilde{\b \beta}$ and $\tilde{\b \Gamma}$, respectively.

In order to overcome some limitations in the estimation of the parameters of the latent process,
we also consider an improved version of the 3S approach, termed as 3S-IMP, in which the second and the third steps are
iterated until convergence, while keeping the results from the first step fixed.
More in detail, after an initial estimate of the latent structure parameters, each Step 2 is based on computing
\begin{eqnarray*}
\tilde{b}_{iu}^{(t)}&\propto&\tilde{p}(U_i^{(t)}=u|\b x_i^{(t)})\;
\hat{p}(\b Y_i^{(t)}=\b y_i^{(t)}|U_i^{(t)}=u),\quad t=1,\ldots,T,\\
&&\\
\tilde{b}_{iuv}^{(t)}&\propto&\tilde{p}(U_i^{(t-1)}=u|\b x_i^{(t-1)})\;\hat{p}(\b Y_i^{(t-1)}=\b y_i^{(t-1)}|U_i^{(t-1)}=u)\\
&&\times\: \tilde{p}(U_i^{(t)}=v|U_i^{(t-1)}=u,\b x_i^{(t)})\;\hat{p}(\b Y_i^{(t)}=\b y_i^{(t)}|U_i^{(t)}=v),
\quad t=2,\ldots,T.
\end{eqnarray*}
In the above expressions, the probabilities $\tilde{p}(\cdot|\cdot)$ are based on the estimates from the previous iteration. Finally, each Step 3
is performed in the same way as in the initial version of the algorithm.

As we shall see in more detail in the following section, the proposed 3S estimation algorithm produces consistent estimates
of the conditional response probabilities $\hat{\phi}_{jy|u}$, with easily interpretable results.  
Moreover, since the first step is based on fitting a basic LC model, several starting values may be easily tried and the global maximum of its likelihood may be found with a reasonable effort. 
Another advantage with respect to the full likelihood method 
is that estimation of $\pi_u$ and $\pi_{v|u}$, or $\b \beta $ and $\b \Gamma$ when we include individual 
covariates, is fast and there are no problems of multiple solutions. Finally, it is worth noting
that the behavior of the estimator of the parameters of the latent process is expected
to improve as the number of response variables increases, since more information 
about clustering is provided by the response variables.  
 
\subsection{Properties of the three-step estimator}

In the following we discuss the asymptotic properties of the 3S estimator that hold as the sample size and the number of response variables tend to infinity,
keeping the number of latent states fixed.  

\begin{prop}\label{prop:1}
When the sample size $n$ tends to infinity, the three-step estimation method described in Section \ref{sec:three-step} produces consistent
estimates of the conditional response probabilities  $\hat{\phi}_{jy|u}$, provided that the data
are generated from a LM model with a given number of states $k$.
\end{prop}

\begin{proof}
Recall that, for the basic LM model for multivariate data, the manifest distribution of each 
vector of response variables $\b Y^{(t)}$ for the $t$-th time occasion may be written as follows:
\begin{eqnarray}\label{eq:dist_proof}\nonumber
p(\b Y^{(t)} = \b y^{(t)}) &=& \sum_{\bl u^{(t)}} \phi_{\bl y^{(t)}|u^{(t)}} \pi_{u^{(1)}} \pi_{u^{(2)}|u^{(1)}}\ldots \pi_{u^{(t)}|u^{(t-1)}}\\ \nonumber
& = &  \sum_{\bl u^{(t)}} \phi_{\bl y^{(t)}|u^{(t)}} \lambda_{u^{(t)}}^{(t)}\\
& = &  \sum_{\bl u^{(t)}} \left[\prod_{j=1}^r \phi_{j y_j^{(t)}|u^{(t)}}\right] \lambda_{u^{(t)}}^{(t)},
\end{eqnarray}
where \[\lambda_{u^{(t)}}^{(t)}= P(U^{(t)}=u^{(t)})= \pi_{u^{(1)}}\prod_{s=2}^t \pi_{u^{(s)}|u^{(s-1)}}.\]
Consider also an alternative data generating process in which, at the first step, we randomly choose the 
time occasion $t$ and, at the second step, we obtain $\tilde{\b y}$ as the response configuration of the sample unit 
corresponding to the selected occasion $t$.  Then, for an arbitrary $t$ we have
 \[
p(\tilde{\b y}) =  \sum_{u} \left[\prod_{j=1}^r \phi_{j y_j|u}\right]\bar{\lambda}_{u},
\]
where $\bar{\lambda}_u = 1/T \sum_t \lambda_{u^{(t)}}^{(t)}$ may be interpreted as the average probability,
overall the different time occasions, of being in state $u$. Therefore, by fitting the basic LC model
for the set of response variables in which the responses provided by the same subject at different occasions are
considered as separated sample units, corresponds to consider a likelihood for all vectors $\tilde{\b y}$, 
whose maximization allows us to obtain consistent estimates of the conditional response probabilities $\phi_{jy|u}$.

When we deal with the extended LM model with individual covariates, expression (\ref{eq:dist_proof}) becomes
\begin{eqnarray*}
p(\b Y^{(t)} = \b y^{(t)}) &=& \int_{\b x} p(\b Y^{(t)} = \b y^{(t)}|\b x) f(\b x) \mathrm{d}\b x \\
&=& \int_{\b x} \left[\sum_{\bl u^{(t)}} p(\b Y^{(t)}=\b y^{(t)}| U^{(t)}=u^{(t)})P(U^{(t)}=u^{(t)}| \b x)\right] f(\b x) \mathrm{d}\b x \\
&=& \int_{\b x} \left[\sum_{\bl u^{(t)}} \phi_{\bl y^{(t)}|u^{(t)}} P(U^{(t)}=u^{(t)}| \b x)\right] f(\b x) \mathrm{d}\b x \\
& = &  \sum_{\bl u^{(t)}} \phi_{\bl y^{(t)}|u^{(t)}}  \int_{\b x} P(U^{(t)}=u^{(t)}| \b x) f(\b x) \mathrm{d}\b x\\
& = & \sum_{\bl u^{(t)}} \phi_{\bl y^{(t)}|u^{(t)}} P(U^{(t)}=u^{(t)}) =  \sum_{\bl u^{(t)}} \left[\prod_{j=1}^r \phi_{j y_j^{(t)}|u^{(t)}}\right] \lambda_{u^{(t)}}^{(t)}.
\end{eqnarray*}
We then conclude that, even in the presence of individual covariates, as $n$ goes to infinity,
the 3S estimator of $\phi_{jy|u}$ is consistent.
\end{proof}

\begin{prop}
When the sample size $n$ and the number of response variables $r$ tend to infinity, the three-step estimation method described in Section \ref{sec:three-step} produces consistent
estimates of the parameters of the latent process, provided that the data
are generated from an LM model with a given number of states $k$, and that exists a fixed constant $\epsilon > 0$ such that $\sum_y(\phi_{jy|u}-\phi_{jy|v})^2\geq \epsilon$, $\forall \; v\neq u$, $\forall \; j$.
\end{prop}

\begin{proof}
Consider that, at Step 2 of the proposed 3S estimation
method, we can assign the subjects to the different latent states for each time occasion on the
basis of the estimated posterior probabilities, that may be expressed as
\begin{eqnarray*}
p(U_i^{(t)} = u|\b Y_i^{(t)}=\b y_i^{(t)}) &= &\frac{p(\b Y_i^{(t)}=\b y_i^{(t)}| U_i^{(t)}=u)\,p(U_i^{(t)}=u)}{p(\b Y_i^{(t)}=\b y_i^{(t)})}\\
  									& = & \frac{ \left(\prod_{j=1}^r \phi_{j y_{ij}^{(t)}|u}\right)p(U_i^{(t)}=u)}{p(\b Y_i^{(t)}=\b y_i^{(t)})}.
\end{eqnarray*}
When the sample size $n$ and the number of response variables $r$ tend to infinity, provided that these variables are not independent 
of the latent variable
defining the latent states, the allocation of the subject
tends to the true one with probability which tends to 1 and, therefore, the parameters of the
latent process are consistently estimated.

\end{proof}

\section{Simulation study}\label{sec:5}

In the following, we rely on a Monte Carlo simulation study aimed at evaluating
the behavior of the proposed 3S and 3S-IMP estimation algorithms, illustrated in Section \ref{sec:three-step}, in comparison with the FML method under a benchmark design and other possible scenarios.
We consider both estimation of the basic LM model and of the extended LM model with individual covariates.
In particular, we generate sample data with certain population parameters and we evaluate the 
performance of the two estimation methods in terms of bias, standard error, and root mean square error
on the basis of 100 simulated samples. 

\subsection{Basic LM model}\label{sec:5basic}

For the basic LM model, the benchmark design (Scenario 1) is based on the following setting:
\begin{itemize}
\item sample size: $n=500$;
\item number of time occasions: $T=5$;
\item number of response variables: $r=5,10,20,50$;
\item number of categories of the response variables: $c_j=2$, $j=1,\ldots,r$;
\item number of latent states: $k=2$,
\item  conditional response probability matrix: 
$\b \Phi_j = \begin{pmatrix} 
                          0.7  & 0.3    \\  
                          0.3 &  0.7  
                          \end{pmatrix} \quad j=1,\ldots,r $;
\item initial probability vector $\b \pi = \begin{pmatrix} 0.5 \\ 0.5 \end{pmatrix}$;
\item transition probability matrix $\b \Pi = \begin{pmatrix} 0.9 & 0.1 \\
															   0.1 & 0.9 
															   \end{pmatrix} $.
\end{itemize}
All the algorithms have been implemented by means of a series of {\sc R}  functions with the recursions implemented in
Fortran (gfortran compiler) and were run on a Intel Core i5 processor with 2.5 GHz.  

In Table \ref{tab:est_Phi} we report the estimation results of the conditional response probabilities 
obtained by the FML method and the 3S estimation method.  
The table only shows the scenario in which $r=5$, because for $r=10,20,50$ we obtained similar results.

\begin{table}[ht]\centering\vspace*{0.25cm}
\small{
\begin{tabular}{rrr|rr|rr}
  \toprule
  \multicolumn7c{FML}\\
      & \multicolumn2c{bias($\hat{\phi}_{j1|u}$)}& \multicolumn2c{se($\hat{\phi}_{j1|u}$)}& \multicolumn2c{rmse($\hat{\phi}_{j1|u}$)} \\   
\midrule
\multicolumn1c{$j$}  & \multicolumn1c{$u=1$} & \multicolumn1{c|}{$u=2$} & \multicolumn1c{$u=1$} & \multicolumn1{c|}{$u=2$}& \multicolumn1c{$u=1$} & \multicolumn1c{$u=2$} \\ 
\midrule
  1 & 0.0011 & -0.0017 & 0.0149 & 0.0154 & 0.0149 & 0.0155 \\ 
 2 & 0.0008 & -0.0031 & 0.0152 & 0.0155 & 0.0152 & 0.0158 \\ 
 3 & -0.0030 & -0.0015 & 0.0154 & 0.0146 & 0.0157 & 0.0147 \\ 
  4 & -0.0002 & 0.0005 & 0.0148 & 0.0145 & 0.0148 & 0.0145 \\ 
  5 & -0.0014 & 0.0012 & 0.0148 & 0.0142 & 0.0148 & 0.0143 \\ 
\toprule
    \multicolumn7c{3S and 3S-IMP}\\
     &  \multicolumn2c{bias($\hat{\phi}_{j1|u}$)}& \multicolumn2c{se($\hat{\phi}_{j1|u}$)}& \multicolumn2c{rmse($\hat{\phi}_{j1|u}$)} \\   
\midrule
\multicolumn1c{$j$}  & \multicolumn1c{$u=1$} & \multicolumn1{c|}{$u=2$} & \multicolumn1c{$u=1$} & \multicolumn1{c|}{$u=2$} & \multicolumn1c{$u=1$} & \multicolumn1c{$u=2$} \\ 
\midrule
 1   & -0.0020 & -0.0042 & 0.0201 & 0.0211 & 0.0202 & 0.0216 \\ 
 2   & -0.0030 & -0.0051 & 0.0203 & 0.0206 & 0.0205 & 0.0213 \\ 
 3   & -0.0071 & -0.0034 & 0.0204 & 0.0200 & 0.0216 & 0.0203 \\ 
 4  & -0.0024 & -0.0032 & 0.0213 & 0.0197 & 0.0215 & 0.0200 \\ 
 5  & -0.0038 & -0.0023 & 0.0202 & 0.0191 & 0.0205 & 0.0193 \\ 
\bottomrule
\end{tabular}\vspace*{0.2cm}
\caption[Text excluding the matrix]{\em Bias, standard error (se) and root mean square error (rmse) of the full maximum likelihood and the three-step estimators of the conditional response probabilities $\phi_{j1|u}$ under Scenario 1, with $r=5$ and $\b \Phi_j=\begin{pmatrix}  0.7  & 0.3    \\  0.3 &  0.7  \end{pmatrix} $}\label{tab:est_Phi}} 
\end{table}
These results are in agreement with Proposition \ref{prop:1}, according to which the conditional response probabilities are consistently estimated
by the proposed  approach, even if some loss of efficiency is observed. 

Table \ref{tab:est_pi} concerns estimators FML, 3S and 3S-IMP of the initial probabilities, whereas 
in Table \ref{tab:est_PI} we display the estimation results of the transition probabilities obtained by
the methods under comparison.  

\begin{table}[ht]\centering\vspace*{0.25cm}
\small{
\begin{tabular}{rr|r|r}
\toprule
 \multicolumn4c{FML}\\
\multicolumn1{c}{$r$} & \multicolumn1{c}{bias($\hat{\pi}_1$)} & \multicolumn1{c}{se($\hat{\pi}_1$)} & \multicolumn1{c}{rmse($\hat{\pi}_1$)}   \\
\midrule
 5 & 0.0025 & 0.0289 & 0.0290 \\  
10 & 0.0013 & 0.0239 & 0.0239 \\ 
20 & 0.0011 & 0.0218 & 0.0219 \\ 
50 &0.0014 & 0.0245 & 0.0245 \\ 
\toprule
   \multicolumn4{c}{3S} \\
\multicolumn1{c}{$r$}  & \multicolumn1{c}{bias($\tilde{\pi}_1$)} & \multicolumn1{c}{se($\tilde{\pi}_1$)} & \multicolumn1{c}{rmse($\tilde{\pi}_1$)}\\
\midrule 
5 & 0.0048 & 0.0328 & 0.0332 \\ 
10 & 0.0006 & 0.0234 & 0.0234 \\
 20 & 0.0015 & 0.0213 & 0.0213 \\
 50 & 0.0015 & 0.0245 & 0.0245 \\ 
\toprule
  \multicolumn4{c}{3S-IMP} \\
& \multicolumn1{c}{bias($\tilde{\pi}_1$)} & \multicolumn1{c}{se($\tilde{\pi}_1$)} & \multicolumn1{c}{rmse($\tilde{\pi}_1$)}  \\
\midrule
5& 0.0029 & 0.0367 & 0.0368 \\ 
10& 0.0008 & 0.0259 & 0.0260 \\
20& 0.0016 & 0.0220 & 0.0220 \\ 
50& 0.0015 & 0.0245 & 0.0246 \\
  \bottomrule
\end{tabular}\vspace*{0.2cm}
\caption[Text excluding the matrix]{\em Bias, standard error (se) and root mean square error (rmse) of the full maximum likelihood and the three-step (standard and improved) estimators of the initial probabilities $\pi_1$ under Scenario 1 with $\b \pi = \begin{pmatrix} 0.5 \\ 0.5 \end{pmatrix}$ }\label{tab:est_pi}} 
\end{table}
\begin{table}[ht]\centering\vspace*{0.25cm}
\small{
\begin{tabular}{cccr|rr|rr}
\toprule
 \multicolumn8c{FML}\\
  && \multicolumn2{c}{bias($\hat{\pi}_{v|u}$)} & \multicolumn2{c}{se($\hat{\pi}_{v|u}$)}  & \multicolumn2{c}{rmse($\hat{\pi}_{v|u}$)}  \\
\midrule
\multicolumn1{c}{$r$} &\multicolumn1{c}{$u$}	& $v= 1$	&	$v=2$	& 	$v=1$	&	$v=2$	& 	$v=1$	&	$v=2$\\
\midrule
 \multirow2*{5} &1	&0.0007 & -0.0007 & 0.0156 & 0.0156 & 0.0156 & 0.0156 \\ 
 & 2 & -0.0014 & 0.0014 & 0.0147 & 0.0147 & 0.0147 & 0.0147 \\ 
 & & & & & & &  \\								
 \multirow2*{10}&1	&-0.0013 & 0.0013 & 0.0125 & 0.0125 & 0.0126 & 0.0126 \\ 
  &2 & 0.0006 & -0.0006 & 0.0125 & 0.0125 & 0.0125 & 0.0125 \\ 
 & & & & & & &  \\										
 \multirow2*{20}&1	& -0.0013 & 0.0013 & 0.0090 & 0.0090 & 0.0091 & 0.0091 \\ 
  &2 & 0.0013 & -0.0013 & 0.0108 & 0.0108 & 0.0109 & 0.0109 \\ 
 & & & & & & &  \\										
 \multirow2*{50}&1	&-0.0001 & 0.0001 & 0.0084 & 0.0084 & 0.0084 & 0.0084 \\ 
  &2 & 0.0002 & -0.0002 & 0.0102 & 0.0102 & 0.0102 & 0.0102 \\ 
\toprule
\multicolumn8c{3S}\\
  && \multicolumn2{c}{bias($\tilde{\pi}_{v|u}$)} & \multicolumn2{c}{se($\tilde{\pi}_{v|u}$)}  & \multicolumn2{c}{rmse($\tilde{\pi}_{v|u}$)}  \\
\midrule
\multicolumn1{c}{$r$} &\multicolumn1{c}{$u$}	& $v= 1$	&	$v=2$	& 	$v=1$	&	$v=2$	& 	$v=1$	&	$v=2$\\
\midrule
 \multirow2*{5} &1	&  -0.2997 & 0.2997 & 0.0273 & 0.0273 & 0.3009 & 0.3009 \\ 
 & 2 & 0.2881 & -0.2881 & 0.0269 & 0.0269 & 0.2894 & 0.2894 \\ 
 & & & & & & &  \\											
 \multirow2*{10}&1	& -0.1872 & 0.1872 & 0.0143 & 0.0143 & 0.1877 & 0.1877 \\ 
  &2 & 0.1875 & -0.1875 & 0.0146 & 0.0146 & 0.1880 & 0.1880 \\ 
 & & & & & & &  \\												
 \multirow2*{20}&1	& -0.0717 & 0.0717 & 0.0098 & 0.0098 & 0.0724 & 0.0724 \\ 
  &2 & 0.0714 & -0.0714 & 0.0118 & 0.0118 & 0.0724 & 0.0724 \\ 
 & & & & & & &  \\										
 \multirow2*{50}&1	&-0.0040 & 0.0040 & 0.0086 & 0.0086 & 0.0094 & 0.0094 \\ 
  &2 & 0.0041 & -0.0041 & 0.0103 & 0.0103 & 0.0111 & 0.0111 \\ 
 \toprule
 \multicolumn8c{3S-IMP}\\
  && \multicolumn2{c}{bias($\tilde{\pi}_{v|u}$)} & \multicolumn2{c}{se($\tilde{\pi}_{v|u}$)}  & \multicolumn2{c}{rmse($\tilde{\pi}_{v|u}$)}  \\
\midrule
\multicolumn1{c}{$r$} &\multicolumn1{c}{$u$}	& $v= 1$	&	$v=2$	& 	$v=1$	&	$v=2$	& 	$v=1$	&	$v=2$\\
\midrule
 \multirow2*{5} &1	&-0.1771 & 0.1771 & 0.0214 & 0.0214 & 0.1784 & 0.1784 \\ 
  & 2 & 0.1690 & -0.1690 & 0.0222 & 0.0222 & 0.1704 & 0.1704 \\ 
 & & & & & & &  \\											
 \multirow2*{10}&1	&-0.0865 & 0.0865 & 0.0140 & 0.0140 & 0.0877 & 0.0877 \\ 
  &2 & 0.0871 & -0.0871 & 0.0138 & 0.0138 & 0.0882 & 0.0882 \\ 
 & & & & & & &  \\												
 \multirow2*{20}&1	& -0.0278 & 0.0278 & 0.0094 & 0.0094 & 0.0293 & 0.0293 \\ 
  &2 & 0.0277 & -0.0277 & 0.0113 & 0.0113 & 0.0299 & 0.0299 \\ 
 & & & & & & &  \\												
 \multirow2*{50}&1	&-0.0014 & 0.0014 & 0.0085 & 0.0085 & 0.0086 & 0.0086 \\ 
  & 2 & 0.0016 & -0.0016 & 0.0102 & 0.0102 & 0.0103 & 0.0103 \\ 
 \bottomrule
\end{tabular}\vspace*{0.2cm}
\caption[Text excluding the matrix]{ \em Bias, standard error (se) and root mean square error (rmse) of the full maximum likelihood and the three-step (standard and improved) estimators of the transition probabilities $\pi_{v|u}$ under Scenario 1 with  $\b \Pi = \begin{pmatrix} 0.9 & 0.1 \\
															   0.1 & 0.9 
															   \end{pmatrix} $}\label{tab:est_PI}} 
\end{table}

From the results above we observe that the bias of $\tilde{\pi}_u$ is always negligible and its root mean 
square error decreases as the number of response variables increases. Moreover, very small differences are found between the 
two versions (3S and 3S-IMP) of the proposed estimator. Regarding the three-step estimator of the transition probabilities, we observe
that its bias is higher with few response variables for each time occasion, especially in the standard version (3S), but it decreases with $r$. The same 
happens for the root mean square error. The 3S-IMP version has better performance, with a bias of $\tilde{\pi}_{v|u}$ which becomes negligible 
when $r\geq 20$. When $r=50$ the results of both versions are very similar to those obtained by the FML estimator.

In the second scenario defined in our simulation study, we evaluate how a lower persistence of the Markov chain may
influence the performance of the estimation algorithms under comparison. In fact, we consider the same parameters as
in the benchmark design apart from the transition probability matrix, which is let equal to $$\b \Pi = \begin{pmatrix} 0.6 & 0.4 \\
															   0.4 & 0.6 
															   \end{pmatrix} .$$
In this situation, the conditional response probabilities are again consistently estimated by both the FML and the 3S 
methods, with very similar results in terms of bias and root mean square error than Scenario 1. 
Very similar performances with respect to the benchmark design are also obtained by the estimators of the parameter
$\pi_u$. Therefore, we report here only the results about estimation of the transition probabilities.	

From Table \ref{tab:PI2} we observe that, with a lower level of persistence 
of the Markov chain, the behavior of $\tilde{\pi}_{v|u}$ improves. In particular, we note
smaller differences, in terms of bias, standard error, and root mean square error, between the FML and the 3S estimator and also between the 3S and
the 3S-IMP estimator.

\begin{table}[ht]\centering\vspace*{0.25cm}
\small{
\begin{tabular}{cccr|rr|rr}
\toprule
 \multicolumn8c{FML}\\
  && \multicolumn2{c}{bias($\hat{\pi}_{v|u}$)} & \multicolumn2{c}{se($\hat{\pi}_{v|u}$)}  & \multicolumn2{c}{rmse($\hat{\pi}_{v|u}$)}  \\
\midrule
\multicolumn1{c}{$r$} &\multicolumn1{c}{$u$}	& $v= 1$	&	$v=2$	& 	$v=1$	&	$v=2$	& 	$v=1$	&	$v=2$\\
\midrule
 \multirow2*{5} &1	& 0.0030 & -0.0030 & 0.0317 & 0.0317 & 0.0319 & 0.0319 \\ 
 & 2 & 0.0013 & -0.0013 & 0.0376 & 0.0376 & 0.0377 & 0.0377 \\ 
 & & & & & & &  \\										
 \multirow2*{10}&1	&-0.0050 & 0.0050 & 0.0194 & 0.0194 & 0.0201 & 0.0201 \\ 
  & 2 & 0.0010 & -0.0010 & 0.0229 & 0.0229 & 0.0229 & 0.0229 \\ 
 & & & & & & &  \\											
 \multirow2*{20}&1	& 0.0026 & -0.0026 & 0.0167 & 0.0167 & 0.0169 & 0.0169 \\ 
  & 2 & -0.0012 & 0.0012 & 0.0180 & 0.0180 & 0.0181 & 0.0181 \\ 
 & & & & & & &  \\											
 \multirow2*{50}&1	&-0.0007 & 0.0007 & 0.0147 & 0.0147 & 0.0147 & 0.0147 \\ 
  & 2 & -0.0016 & 0.0016 & 0.0147 & 0.0147 & 0.0148 & 0.0148 \\ 
\toprule
 \multicolumn8c{3S}\\
  && \multicolumn2{c}{bias($\tilde{\pi}_{v|u}$)} & \multicolumn2{c}{se($\tilde{\pi}_{v|u}$)}  & \multicolumn2{c}{rmse($\tilde{\pi}_{v|u}$)}  \\
\midrule
\multicolumn1{c}{$r$} &\multicolumn1{c}{$u$}	& $v= 1$	&	$v=2$	& 	$v=1$	&	$v=2$	& 	$v=1$	&	$v=2$\\
\midrule
 \multirow2*{5} &1	&  -0.0741 & 0.0741 & 0.0293 & 0.0293 & 0.0797 & 0.0797 \\ 
&  2 & 0.0724 & -0.0724 & 0.0321 & 0.0321 & 0.0792 & 0.0792 \\  
 & & & & & & &  \\											
 \multirow2*{10}&1	& -0.0495 & 0.0495 & 0.0160 & 0.0160 & 0.0520 & 0.0520 \\ 
  &  2 & 0.0467 & -0.0467 & 0.0180 & 0.0180 & 0.0500 & 0.0500 \\ 
 & & & & & & &  \\												
 \multirow2*{20}&1	& -0.0150 & 0.0150 & 0.0150 & 0.0150 & 0.0212 & 0.0212 \\ 
  &  2 & 0.0166 & -0.0166 & 0.0165 & 0.0165 & 0.0234 & 0.0234 \\  
 & & & & & & &  \\											
 \multirow2*{50}&1	&-0.0017 & 0.0017 & 0.0146 & 0.0146 & 0.0147 & 0.0147 \\ 
  &  2 & -0.0006 & 0.0006 & 0.0146 & 0.0146 & 0.0146 & 0.0146 \\ 
 \toprule
\multicolumn8c{3S-IMP}\\
  && \multicolumn2{c}{bias($\tilde{\pi}_{v|u}$)} & \multicolumn2{c}{se($\tilde{\pi}_{v|u}$)}  & \multicolumn2{c}{rmse($\tilde{\pi}_{v|u}$)}  \\
 \midrule
\multicolumn1{c}{$r$} &\multicolumn1{c}{$u$}	& $v= 1$	&	$v=2$	& 	$v=1$	&	$v=2$	& 	$v=1$	&	$v=2$\\
\midrule
 \multirow2*{5} &1	&-0.0481 & 0.0481 & 0.0297 & 0.0297 & 0.0565 & 0.0565 \\ 
  &  2 & 0.0480 & -0.0480 & 0.0345 & 0.0345 & 0.0591 & 0.0591 \\  
 & & & & & & &  \\									
 \multirow2*{10}&1	&-0.0304 & 0.0304 & 0.0175 & 0.0175 & 0.0351 & 0.0351 \\ 
  &  2 & 0.0274 & -0.0274 & 0.0202 & 0.0202 & 0.0340 & 0.0340 \\ 
 & & & & & & &  \\											
 \multirow2*{20}&1	& -0.0066 & 0.0066 & 0.0159 & 0.0159 & 0.0172 & 0.0172 \\ 
  &  2 & 0.0080 & -0.0080 & 0.0173 & 0.0173 & 0.0191 & 0.0191 \\ 
 & & & & & & &  \\											
 \multirow2*{50}&1	&-0.0012 & 0.0012 & 0.0147 & 0.0147 & 0.0147 & 0.0147 \\ 
  &  2 & -0.0011 & 0.0011 & 0.0147 & 0.0147 & 0.0147 & 0.0147 \\ 
   \bottomrule
\end{tabular}\vspace*{0.2cm}
\caption[Text excluding the matrix]{\em Bias, standard error (se) and root mean square error (rmse) of the full maximum likelihood and the three-step (standard and improved) estimators of the transition probabilities $\pi_{v|u}$ under Scenario 2 with $\b \Pi = \begin{pmatrix} 0.6 & 0.4 \\
															   0.4 & 0.6 
															   \end{pmatrix} $ }\label{tab:PI2}} 
\end{table}

Under Scenario 3 we consider more separated latent states, with a parameter setting
equal to the benchmark design, apart from the conditional response probability matrix, that is,
\[\b \Phi_j = \begin{pmatrix} 
                          0.9  & 0.1    \\  
                          0.1 &  0.9  
                          \end{pmatrix}, \quad j=1,\ldots,r.\]
In this context, all the estimators improve their performances, leading again to consistent estimates of the conditional
response probabilities $\phi_{jy|u}$; we also observe an improvement in estimating the initial probabilities $\pi_u$. Focusing only on the transition probabilities, whose estimation is the most challenging, from Table \ref{tab:PI3} we observe a great improvement of the 3S
approach, with negligible differences in terms of bias and root mean square error of $\tilde{\pi}_{v|u}$, with respect
to the FML method and between the two versions of the three-step method, especially when $r\geq 10$.

\begin{table}[ht]\centering\vspace*{0.25cm}
\small{
\begin{tabular}{cccr|rr|rr}
\toprule
 \multicolumn8c{FML}\\
  && \multicolumn2{c}{bias($\hat{\pi}_{v|u}$)} & \multicolumn2{c}{se($\hat{\pi}_{v|u}$)}  & \multicolumn2{c}{rmse($\hat{\pi}_{v|u}$)}  \\
\midrule
\multicolumn1{c}{$r$} &\multicolumn1{c}{$u$}	& $v= 1$	&	$v=2$	& 	$v=1$	&	$v=2$	& 	$v=1$	&	$v=2$\\
\midrule
 \multirow2	*{5} &1	&  -0.0023 & 0.0023 & 0.0106 & 0.0106 & 0.0109 & 0.0109 \\ 
  &2 & 0.0010 & -0.0010 & 0.0101 & 0.0101 & 0.0101 & 0.0101 \\ 
 & & & & & & &  \\									
 \multirow2*{10}&1	&-0.0019 & 0.0019 & 0.0102 & 0.0102 & 0.0104 & 0.0104 \\ 
  &2  & -0.0002 & 0.0002 & 0.0097 & 0.0097 & 0.0097 & 0.0097 \\  
 & & & & & & &  \\										
 \multirow2*{20}&1	&0.0004 & -0.0004 & 0.0103 & 0.0103 & 0.0103 & 0.0103 \\ 
  &2  & -0.0015 & 0.0015 & 0.0094 & 0.0094 & 0.0095 & 0.0095 \\ 
 & & & & & & &  \\									
 \multirow2*{50}&1	&0.0004 & -0.0004 & 0.0083 & 0.0083 & 0.0084 & 0.0084 \\ 
  &2  & 0.0002 & -0.0002 & 0.0084 & 0.0084 & 0.0084 & 0.0084 \\ 
\toprule
 \multicolumn8c{3S}\\
  && \multicolumn2{c}{bias($\tilde{\pi}_{v|u}$)} & \multicolumn2{c}{se($\tilde{\pi}_{v|u}$)}  & \multicolumn2{c}{rmse($\tilde{\pi}_{v|u}$)}  \\
\midrule
\multicolumn1{c}{$r$} &\multicolumn1{c}{$u$}	& $v= 1$	&	$v=2$	& 	$v=1$	&	$v=2$	& 	$v=1$	&	$v=2$\\
\midrule
 \multirow2*{5} &1	& -0.0265 & 0.0265 & 0.0108 & 0.0108 & 0.0286 & 0.0286 \\ 
 & 2 & 0.0254 & -0.0254 & 0.0102 & 0.0102 & 0.0274 & 0.0274 \\ 
 & & & & & & &  \\									
 \multirow2*{10}&1	& -0.0035 & 0.0035 & 0.0104 & 0.0104 & 0.0109 & 0.0109 \\ 
  & 2 & 0.0014 & -0.0014 & 0.0097 & 0.0097 & 0.0098 & 0.0098 \\ 
 & & & & & & &  \\										
 \multirow2*{20}&1	& 0.0004 & -0.0004 & 0.0103 & 0.0103 & 0.0104 & 0.0104 \\ 
  & 2 & -0.0015 & 0.0015 & 0.0094 & 0.0094 & 0.0095 & 0.0095 \\ 
 & & & & & & &  \\										
 \multirow2*{50}&1	&0.0004 & -0.0004 & 0.0083 & 0.0083 & 0.0084 & 0.0084 \\ 
  & 2 & 0.0002 & -0.0002 & 0.0084 & 0.0084 & 0.0084 & 0.0084 \\
 \toprule
 \multicolumn8c{3S-IMP}\\
  && \multicolumn2{c}{bias($\tilde{\pi}_{v|u}$)} & \multicolumn2{c}{se($\tilde{\pi}_{v|u}$)}  & \multicolumn2{c}{rmse($\tilde{\pi}_{v|u}$)}  \\
\midrule
\multicolumn1{c}{$r$} &\multicolumn1{c}{$u$}	& $v= 1$	&	$v=2$	& 	$v=1$	&	$v=2$	& 	$v=1$	&	$v=2$\\
\midrule
 \multirow2*{5} &1	&-0.0116 & 0.0116 & 0.0108 & 0.0108 & 0.0158 & 0.0158 \\ 
 & 2 & 0.0104 & -0.0104 & 0.0101 & 0.0101 & 0.0145 & 0.0145 \\ 
 & & & & & & &  \\								
 \multirow2*{10}&1	&-0.0024 & 0.0024 & 0.0102 & 0.0102 & 0.0105 & 0.0105 \\ 
  & 2 & 0.0003 & -0.0003 & 0.0097 & 0.0097 & 0.0097 & 0.0097 \\ 
 & & & & & & &  \\										
 \multirow2*{20}&1	&0.0004 & -0.0004 & 0.0103 & 0.0103 & 0.0104 & 0.0104 \\ 
& 2 & -0.0015 & 0.0015 & 0.0094 & 0.0094 & 0.0095 & 0.0095 \\
 & & & & & & &  \\											
 \multirow2*{50}&1	&0.0004 & -0.0004 & 0.0083 & 0.0083 & 0.0084 & 0.0084 \\ 
  & 2 & 0.0002 & -0.0002 & 0.0084 & 0.0084 & 0.0084 & 0.0084 \\ 
   \bottomrule
\end{tabular}\vspace*{0.2cm}
\caption[Text excluding the matrix]{\em Bias, standard error (se) and root mean square error (rmse) of the full maximum likelihood and the three-step (standard and improved) estimators of the transition probabilities $\pi_{v|u}$ under Scenario 3 with $\b \Pi = \begin{pmatrix} 0.9 & 0.1 \\
															   0.1 & 0.9 
															   \end{pmatrix}$ }\label{tab:PI3}} 
\end{table}

The fourth scenario we consider for the basic LM model allows us to evaluate 
the influence of a larger number of latent states. We then apply the following setting:
\begin{itemize} 
\item number of latent states: $k=3$;
\item conditional response probability matrix: $\b \Phi_j =  \begin{pmatrix}
0.7 &  0.3 & 0.5 \\
0.3 & 0.7 & 0.5
\end{pmatrix},\quad j=1,\ldots,r$;
\item initial probability vector: $\b\pi=\begin{pmatrix}1/3\cr 1/3 \cr 1/3\end{pmatrix}$;
\item transition probability matrix: $ \b \Pi =  \begin{pmatrix}
0.6 &  0.2  & 0.2\\
0.2 & 0.6 & 0.2\\
0.2 & 0.2 & 0.6
\end{pmatrix}$.
\end{itemize}
The remaining parameters are left unchanged with respect to the benchmark design.

In such a context the performances of all estimators get worse, with higher bias and root mean square error;
see Table \ref{tab:PI4} which reports the results regarding the estimation of $\pi_{v|u}$.
As in the previous scenarios, the performance of the proposed estimator improves as the number of response variables increases.

\begin{table}[ht]\centering\vspace*{0.25cm}
\small{
\begin{tabular}{ccrrr|rrr|rrr}
\toprule
 \multicolumn{11}c{FML}\\
  && \multicolumn3{c}{bias($\hat{\pi}_{v|u}$)} & \multicolumn3{c}{se($\hat{\pi}_{v|u}$)}  & \multicolumn3{c}{rmse($\hat{\pi}_{v|u}$)}  \\
\midrule
\multicolumn1{c}{$r$} &\multicolumn1{c}{$u$}	& $v= 1$	&	$v=2$	& $v= 3$&	$v=1$	&	$v=2$ & $v= 3$	& 	$v=1$	&	$v=2$& $v= 3$\\
\midrule
 \multirow3*{5} &1	&  -0.080 & 0.044 & 0.036 & 0.257 & 0.240 & 0.156 & 0.270 & 0.244 & 0.160 \\ 
 & 2 & -0.060 & 0.033 & 0.026 & 0.151 & 0.211 & 0.227 & 0.163 & 0.214 & 0.229 \\ 
  &3 & -0.004 & 0.041 & -0.037 & 0.151 & 0.163 & 0.174 & 0.151 & 0.168 & 0.178 \\ 
 & & & & & & & & & & \\									
 \multirow3*{10}&1	&-0.095 & 0.053 & 0.042 & 0.247 & 0.171 & 0.149 & 0.265 & 0.179 & 0.155 \\ 
 & 2 & 0.040 & -0.114 & 0.074 & 0.181 & 0.245 & 0.216 & 0.186 & 0.270 & 0.228 \\ 
  & 3 & 0.017 & 0.045 & -0.062 & 0.154 & 0.203 & 0.214 & 0.155 & 0.208 & 0.222 \\ 
 & & & & & & & & & & \\										
 \multirow3*{50}&1	&  -0.002 & 0.003 & -0.001 & 0.023 & 0.021 & 0.017 & 0.023 & 0.021 & 0.017 \\ 
  & 2 & 0.001 & 0.001 & -0.002 & 0.021 & 0.028 & 0.016 & 0.021 & 0.028 & 0.017 \\ 
  & 3 & -0.003 & 0.004 & -0.001 & 0.016 & 0.019 & 0.022 & 0.017 & 0.019 & 0.022 \\ 
\toprule
\multicolumn{11}c{3S}\\
  && \multicolumn3{c}{bias($\tilde{\pi}_{v|u}$)} & \multicolumn3{c}{se($\tilde{\pi}_{v|u}$)}  & \multicolumn3{c}{rmse($\tilde{\pi}_{v|u}$)}  \\
  \midrule
\multicolumn1{c}{$r$} &\multicolumn1{c}{$u$}	& $v= 1$	&	$v=2$	& $v= 3$&	$v=1$	&	$v=2$ & $v= 3$	& 	$v=1$	&	$v=2$& $v= 3$\\
\midrule
 \multirow3*{5} &1	& -0.275 & 0.199 & 0.076 & 0.131 & 0.138 & 0.143 & 0.305 & 0.242 & 0.162 \\ 
 & 2 & 0.094 & -0.185 & 0.091 & 0.116 & 0.142 & 0.142 & 0.149 & 0.233 & 0.168 \\ 
  &3 & 0.076 & 0.207 & -0.283 & 0.117 & 0.150 & 0.157 & 0.140 & 0.255 & 0.324 \\ 
 & & & & & & & & & & \\									
 \multirow3*{10}&1	& -0.178 & 0.095 & 0.083 & 0.115 & 0.114 & 0.098 & 0.212 & 0.148 & 0.129 \\ 
  &2 & 0.155 & -0.288 & 0.134 & 0.085 & 0.116 & 0.096 & 0.176 & 0.311 & 0.165 \\ 
  &3 & 0.110 & 0.099 & -0.208 & 0.103 & 0.117 & 0.127 & 0.150 & 0.153 & 0.244 \\ 
 & & & & & & & & & & \\									
 \multirow3*{50}&1	&-0.077 & 0.071 & 0.006 & 0.020 & 0.016 & 0.014 & 0.080 & 0.073 & 0.015 \\ 
  &2& 0.075 & -0.145 & 0.070 & 0.016 & 0.025 & 0.014 & 0.076 & 0.147 & 0.072 \\ 
  &3 & 0.003 & 0.073 & -0.076 & 0.013 & 0.017 & 0.020 & 0.013 & 0.075 & 0.079 \\
 \toprule
\multicolumn{11}c{3S-IMP}\\
  && \multicolumn3{c}{bias($\tilde{\pi}_{v|u}$)} & \multicolumn3{c}{se($\tilde{\pi}_{v|u}$)}  & \multicolumn3{c}{rmse($\tilde{\pi}_{v|u}$)}  \\
  \midrule
\multicolumn1{c}{$r$} &\multicolumn1{c}{$u$}	& $v= 1$	&	$v=2$	& $v= 3$&	$v=1$	&	$v=2$ & $v= 3$	& 	$v=1$	&	$v=2$& $v= 3$\\
\midrule
 \multirow3*{5} &1	& -0.236 & 0.183 & 0.053 & 0.142 & 0.126 & 0.133 & 0.275 & 0.222 & 0.143 \\ 
 & 2 & 0.088 & -0.172 & 0.083 & 0.105 & 0.140 & 0.130 & 0.137 & 0.222 & 0.155 \\ 
 & 3 & 0.045 & 0.205 & -0.249 & 0.115 & 0.164 & 0.167 & 0.123 & 0.262 & 0.300 \\ 
 & & & & & & & & & & \\								
 \multirow3*{10}&1	&-0.135 & 0.092 & 0.043 & 0.120 & 0.120 & 0.087 & 0.181 & 0.151 & 0.097 \\ 
 & 2 & 0.142 & -0.270 & 0.128 & 0.074 & 0.128 & 0.092 & 0.160 & 0.299 & 0.157 \\ 
 & 3 & 0.075 & 0.086 & -0.160 & 0.094 & 0.119 & 0.139 & 0.120 & 0.146 & 0.212 \\ 
 & & & & & & & & & & \\										
 \multirow3*{50}&1	&-0.042 & 0.042 & 0.001 & 0.022 & 0.020 & 0.015 & 0.048 & 0.046 & 0.015 \\ 
  & 2 & 0.043 & -0.083 & 0.040 & 0.019 & 0.028 & 0.015 & 0.047 & 0.088 & 0.043 \\ 
  & 3 & -0.003 & 0.045 & -0.042 & 0.014 & 0.019 & 0.021 & 0.014 & 0.049 & 0.047 \\
   \bottomrule
\end{tabular}\vspace*{0.2cm}
\caption[Text excluding the matrix]{\em Bias, standard error (se) and root mean square error (rmse) of the full maximum likelihood and the three-step (standard and improved) estimators of the transition probabilities $\pi_{v|u}$ under Scenario 4 with  $ \b \Pi =  \begin{pmatrix}
0.6 &  0.2  & 0.2\\
0.2 & 0.6 & 0.2\\
0.2 & 0.2 & 0.6
\end{pmatrix}$}\label{tab:PI4}} 
\end{table}

We finally assess the influence of the sample size $n$ and of the number of time occasions $T$ on 
the estimation results by considering the same parameters setting as in Scenario 1 apart from $n=1000$ and
$T=8$. When the sample size and/or the number of time occasion increase, we observe an expected improvement of all estimators;
for example, the standard error of the estimators decreases, in general, with $\sqrt{n}$. We do not report here these results for reason of space. 

In comparing the estimation approaches, it is relevant to take into account the computing time required
to run the algorithms of interest. 
%
In particular, we observe that under the benchmark design (Scenario 1) the computing time required by the 3S approach is,
in all samples, significantly lower than that required by the FML approach. On the other 
hand, the computational cost of the 3S-IMP version is higher, especially with
a few response variables. 
The number of iterations required by the 3S-IMP version to reach the convergence goes, in average, from 27 (with $r=5$)  to 4 (with $r=50$). 
Under Scenario 2, we observe an increase of 
the computational cost of the FML estimation algorithm. 
On the other hand, the computing time required by the two versions of the three-step algorithm is similar to that of Scenario 1, with a very similar number of iterations to reach the convergence
(from 21 with $r=5$, to 4 with $r=50$).
The third scenario shows a strong decrease of the computing time for all the algorithms,
together with a reduction of the number of iterations to converge required by the 3S-IMP algorithm, 
which goes from 6 (with $r=5$) to 1 (with $r=50$). In this context the computational cost of the 3S algorithm remains
always significantly lower than that of the FML algorithm, whereas, for the 3S-IMP algorithm, this happens only for $r \geq 20$.
Finally, when the number of latent states increases (Scenario 4), the computing time required by the FML algorithm 
becomes significantly higher. We also observe a strong increase of the number of 
iterations required by the 3S-IMP approach to reach the convergence, which goes from 95 (with $r=5$) to 
15 (with $r=50$).

\subsection{LM model with individual covariates}\label{sec:5cov}

For the extended LM model with individual covariates illustrated in Section \ref{sec:lmcov}, the benchmark design, which
corresponds to Scenario 1 of the basic LM model, is based on the following setting:
\begin{itemize}
\item sample size: $n=500$;
\item number of time occasions: $T=5$;
\item number of response variables: $r=5,10,20,50$;
\item number of categories of the response variables: $c_j=2$, $j=1,\ldots,r$;
\item number of latent states: $k=2$;
\item number of covariates affecting the initial probabilities: $q_1 = 2$;
\item number of covariates affecting the transition probabilities:  $q_2 = 2$; 
\item  conditional response probability matrix: 
$\b \Phi_j = \begin{pmatrix} 
                          0.7  & 0.3    \\  
                          0.3 &  0.7  
                          \end{pmatrix} \quad j=1,\ldots,r $;
\item $\b \beta = \begin{pmatrix}
					0 & 0.5 & 1
				\end{pmatrix}\tr$;         
\item $\b \Gamma =\begin{pmatrix}
					\log(0.1/0.9) & \log(0.1/0.9) \\
					0.5 & 0.5 \\
					1 & 1 
					\end{pmatrix}$. 
\end{itemize}

As for the basic LM model, the estimation results are based on 100 simulated samples; moreover, the individual covariates 
have been generated from an AR(1) process with autoregressive parameters equal to 0.5 and a gaussian noise process with variance 
equal to 1.

Table \ref{tab:Phicov}, in which are reported
the results of the FML and the 3S estimators of the probabilities $\phi_{jy|u}$ for $r=5$,
confirms that the proposed approach leads again to consistent estimates of the conditional response probabilities.

\begin{table}[ht]\centering\vspace*{0.25cm}
\small{
\begin{tabular}{rrr|rr|rr}
  \toprule
 \multicolumn7c{FML} \\
      & \multicolumn2c{Bias($\hat{\phi}_{j1|u}$)}& \multicolumn2c{se($\hat{\phi}_{j1|u}$)}& \multicolumn2c{rmse($\hat{\phi}_{j1|u}$)} \\   
\midrule
$j$ & $u=1$ & $u=2$ & $u=1$ & $u=2$& $u=1$ & $u=2$ \\ 
\midrule	
 1 & 0.0004 & 0.0018 & 0.0150 & 0.0158 & 0.0150 & 0.0159 \\ 
  2 & 0.0003 & -0.0001 & 0.0182 & 0.0144 & 0.0182 & 0.0144 \\ 
 3 & 0.0003 & 0.0014 & 0.0156 & 0.0174 & 0.0156 & 0.0174 \\ 
 4  & -0.0003 & 0.0011 & 0.0138 & 0.0152 & 0.0138 & 0.0153 \\ 
 5 & -0.0009 & 0.0009 & 0.0159 & 0.0184 & 0.0159 & 0.0184 \\ 
  \toprule
   \multicolumn7c{3S and 3S-IMP} \\
      & \multicolumn2c{Bias($\hat{\phi}_{j1|u}$)}& \multicolumn2c{se($\hat{\phi}_{j1|u}$)}& \multicolumn2c{rmse($\hat{\phi}_{j1|u}$)} \\   
\midrule
$j$& $u=1$ & $u=2$ & $u=1$ & $u=2$& $u=1$ & $u=2$ \\ 
 \midrule	
1 & -0.0020 & 0.0022 & 0.0211 & 0.0226 & 0.0212 & 0.0227 \\ 
 2 & -0.0020 & 0.0001 & 0.0229 & 0.0200 & 0.0229 & 0.0200 \\  
3 & -0.0009 & 0.0005 & 0.0222 & 0.0232 & 0.0222 & 0.0232 \\ 
4& -0.0005 & -0.0012 & 0.0213 & 0.0207 & 0.0213 & 0.0208 \\ 
5& -0.0023 & 0.0000 & 0.0205 & 0.0230 & 0.0206 & 0.0230 \\ 
\bottomrule
\end{tabular}\vspace*{0.2cm}
\caption[Text excluding the matrix]{ \em Bias, standard error (se) and root mean square error (rmse) of the full maximum likelihood and the three-step (standard and improved) estimators of the conditional response probabilities $\phi_{j1|u}$ under Scenario 1 with individual covariates, $r=5$ and $\b \Phi_j = \begin{pmatrix} 
                          0.7  & 0.3    \\  
                          0.3 &  0.7  
                          \end{pmatrix}$ }\label{tab:Phicov}} 
\end{table}

Table \ref{tab:becov} shows the estimation results of parameters in $\b \beta$
obtained by the two methods. From the table, we note that the 3S estimator has a low bias when $r \geq 20$. Moreover, we observe  a much lower bias of the 3S-IMP estimator, whose behavior is similar to that of the FML estimator.
The standard error and the root mean square error decrease as the number of response variables increases. 

\begin{table}[ht]\centering\vspace*{0.25cm}
\small{
\begin{tabular}{cr|r|r}
\toprule
 \multicolumn4{c}{FML} \\
$r$ & \multicolumn1c{bias($\hat{\b \beta}$)} & \multicolumn1c{se($\hat{\b \beta}$)} & \multicolumn1c{rmse($\hat{\b \beta}$)} \\ 
\midrule
 \multirow3*{5}  & 0.0287 & 0.1386 & 0.1415 \\ 
& 0.0096 & 0.1448 & 0.1451 \\ 
& -0.0004 & 0.1610 & 0.1610 \\ 
&&&\\
\multirow3*{10 } & 0.0077 & 0.1197 & 0.1199 \\ 
 & -0.0106 & 0.1246 & 0.1250 \\ 
  & 0.0055 & 0.1266 & 0.1267 \\ 
&&&\\
\multirow3*{20 } & -0.0167 & 0.1029 & 0.1042 \\ 
  & -0.0147 & 0.1137 & 0.1147 \\ 
 & 0.0266 & 0.1247 & 0.1275 \\ 
&&&\\
\multirow3*{50 } & 0.0045 & 0.0938 & 0.0939 \\ 
 & 0.0129 & 0.1049 & 0.1056 \\ 
 & -0.0092 & 0.1185 & 0.1189 \\ 
\toprule
 \multicolumn4{c}{3S} \\
$r$ & \multicolumn1c{bias($\tilde{\b \beta}$)} & \multicolumn1c{se($\tilde{\b \beta}$)} & \multicolumn1c{rmse($\tilde{\b \beta}$)} \\ 
\midrule
\multirow3*{5} & 0.0006 & 0.1689 & 0.1689\\ 
    & -0.2734 & 0.0687 & 0.2819 \\
   & -0.5714 & 0.0641 & 0.5750\\ 
&&&\\
\multirow3*{10 }& -0.0156 & 0.1100 & 0.1111\\ 
   & -0.1818 & 0.0858 & 0.2011\\ 
   & -0.3453 & 0.0820 & 0.3549\\ 
&&&\\
\multirow3*{20 } & -0.0186 & 0.0928 & 0.0947\\ 
   & -0.0778 & 0.1008 & 0.1273\\  
   & -0.1065 & 0.1065 & 0.1507 \\  
&&&\\
 \multirow3*{50 } & 0.0045 & 0.0937 & 0.0938\\  
  & 0.0091 & 0.1035 & 0.1039\\  
 & -0.0181 & 0.1168 & 0.1182\\ 
\toprule
 \multicolumn4{c}{3S-IMP} \\
$r$ & \multicolumn1c{bias($\tilde{\b \beta}$)} & \multicolumn1c{se($\tilde{\b \beta}$)} & \multicolumn1c{rmse($\tilde{\b \beta}$)} \\ 
\midrule
\multirow3*{5} & 0.0438 & 0.2018 & 0.2065 \\ 
  & 0.0191 & 0.1594 & 0.1605 \\ 
& -0.0128 & 0.1737 & 0.1741 \\ 
&&&\\
\multirow3*{10 }& 0.0081 & 0.1331 & 0.1334 \\ 
& -0.0111 & 0.1346 & 0.1350 \\ 
& -0.0004 & 0.1372 & 0.1372 \\ 
&&&\\  
\multirow3*{20 }& -0.0126 & 0.1015 & 0.1023 \\ 
& -0.0147 & 0.1136 & 0.1145 \\
& 0.0288 & 0.1265 & 0.1297 \\
&&&\\
 \multirow3*{50 }& 0.0047 & 0.0942 & 0.0943 \\
 & 0.0133 & 0.1050 & 0.1059 \\
 & -0.0096 & 0.1185 & 0.1189 \\ 
\bottomrule
\end{tabular}\vspace*{0.2cm}
\caption[Text excluding the matrix]{\em Bias, standard error (se) and root mean square error (rmse) of the full maximum likelihood and the three-step (standard and improved) estimators of $\b \beta$ under Scenario 1 with individual covariates and $\b \beta = \begin{pmatrix}
					0 & 0.5 & 1
				\end{pmatrix}\tr$ }\label{tab:becov}
}
\end{table}

Regarding the three-step estimation of parameters in $\b \Gamma$, we observe that the bias (and the root mean square error) may be large with a few response variables for each time occasion 
(see Table \ref{tab:Gacov}), but it decreases with
$r$ and it is negligible for $r=50$. Moreover, we note an 
improvement from the 3S to the 3S-IMP version of the algorithm. 
\begin{table}[ht]\centering\vspace{-2mm}
\small{
\begin{tabular}{crr|rr|rr}
 \toprule
& \multicolumn6c{FML}\\
 $r$& \multicolumn2c{bias($\hat{\bl\Gamma}$)} &  \multicolumn2c{se($\hat{\bl\Gamma}$)} &  \multicolumn2c{rmse($\hat{\bl\Gamma}$)}\\
\midrule
\multirow3*{5}  & -0.0554 & -0.0539 & 0.2323 & 0.2320 & 0.2388 & 0.2382 \\ 
    & 0.0268 & 0.0100 & 0.1545 & 0.1497 & 0.1568 & 0.1501 \\ 
   & 0.0350 & 0.0482 & 0.2018 & 0.1818 & 0.2048 & 0.1881 \\ 
&&&&&&\\
\multirow3*{10}   & -0.0109 & -0.0063 & 0.1652 & 0.1775 & 0.1656 & 0.1777 \\ 
   & 0.0094 & 0.0048 & 0.1173 & 0.1056 & 0.1177 & 0.1057 \\ 
   & 0.0239 & -0.0084 & 0.1330 & 0.1393 & 0.1351 & 0.1396 \\ 
&&&&&&\\
\multirow3*{20}    & -0.0256 & -0.0373 & 0.1403 & 0.1479 & 0.1426 & 0.1526 \\ 
   & 0.0020 & 0.0055 & 0.0857 & 0.0851 & 0.0857 & 0.0853 \\ 
   & 0.0154 & 0.0247 & 0.1091 & 0.1168 & 0.1102 & 0.1194 \\ 
&&&&&&\\
  \multirow3*{50}  &  -0.0145 & -0.0036 & 0.1109 & 0.1387 & 0.1118 & 0.1387 \\ 
   & -0.0018 & 0.0079 & 0.0794 & 0.0962 & 0.0794 & 0.0965 \\ 
   & 0.0148 & -0.0132 & 0.0968 & 0.1030 & 0.0979 & 0.1039 \\ 	
  \toprule	
  & \multicolumn6c{3S}\\
 $r$& \multicolumn2c{bias($\tilde{\bl\Gamma}$)} &  \multicolumn2c{se($\tilde{\bl\Gamma}$)} &  \multicolumn2c{rmse($\tilde{\bl\Gamma}$)}\\
\midrule
\multirow3*{5}  &  1.7805 & 1.8592 & 0.1379 & 0.1369 & 1.7858 & 1.8642 \\ 
   & -0.4244 & -0.4588 & 0.0365 & 0.0369 & 0.4260 & 0.4603 \\ 
   & -0.8548 & -0.9021 & 0.0371 & 0.0355 & 0.8556 & 0.9028 \\ 
&&&&&&\\
\multirow3*{10}   & 1.3428 & 1.4333 & 0.0786 & 0.0742 & 1.3451 & 1.4352 \\ 
   & -0.3462 & -0.3788 & 0.0510 & 0.0433 & 0.3500 & 0.3813 \\ 
   & -0.6885 & -0.7610 & 0.0475 & 0.0524 & 0.6902 & 0.7628 \\ 
&&&&&&\\
\multirow3*{20}    & 0.6996 & 0.7677 & 0.0811 & 0.0865 & 0.7043 & 0.7726 \\ 
   & -0.2023 & -0.2235 & 0.0594 & 0.0576 & 0.2109 & 0.2308 \\ 
   & -0.3990 & -0.4391 & 0.0660 & 0.0707 & 0.4044 & 0.4448 \\ 
&&&&&&\\ 
  \multirow3*{50}  & 0.0453 & 0.0634 & 0.1077 & 0.1318 & 0.1168 & 0.1463 \\ 
 & -0.0219 & -0.0133 & 0.0779 & 0.0929 & 0.0809 & 0.0939 \\ 
  & -0.0203 & -0.0525 & 0.0948 & 0.0983 & 0.0969 & 0.1114 \\ 
   \toprule	
  & \multicolumn6c{3S-IMP}\\
 $r$& \multicolumn2c{bias($\tilde{\bl\Gamma}$)} &  \multicolumn2c{se($\tilde{\bl\Gamma}$)} &  \multicolumn2c{rmse($\tilde{\bl\Gamma}$)}\\
\midrule  
\multirow3*{5}  & 1.2476 & 1.3172 & 0.1288 & 0.1279 & 1.2542 & 1.3234 \\ 
   & -0.3066 & -0.3544 & 0.0801 & 0.0774 & 0.3169 & 0.3628 \\ 
   & -0.6253 & -0.6870 & 0.0911 & 0.0771 & 0.6319 & 0.6913 \\ 
&&&&&&\\
\multirow3*{10}   & 0.7678 & 0.8329 & 0.0985 & 0.0935 & 0.7741 & 0.8382 \\ 
   & -0.2010 & -0.2300 & 0.0859 & 0.0736 & 0.2186 & 0.2415 \\ 
   & -0.3918 & -0.4713 & 0.0830 & 0.0876 & 0.4005 & 0.4793 \\ 
&&&&&&\\
\multirow3*{20}    & 0.2746 & 0.3031 & 0.1148 & 0.1101 & 0.2976 & 0.3224 \\ 
   & -0.0800 & -0.0901 & 0.0724 & 0.0725 & 0.1079 & 0.1156 \\ 
   & -0.1495 & -0.1691 & 0.0946 & 0.0940 & 0.1769 & 0.1935 \\ 
&&&&&&\\
  \multirow3*{50}  & 0.0056 & 0.0186 & 0.1103 & 0.1370 & 0.1105 & 0.1382 \\ 
   & -0.0083 & 0.0009 & 0.0798 & 0.0955 & 0.0802 & 0.0955 \\ 
   & 0.0028 & -0.0261 & 0.0958 & 0.1015 & 0.0958 & 0.1048 \\ 
\bottomrule
\end{tabular}\vspace*{0.2cm}
\caption[Text excluding the matrix]{ \em Bias, standard error (se) and root mean square error (rmse) of the full maximum likelihood and the three-step (standard and improved)  estimators of $\b \Gamma$, under Scenario 1 with individual covariates 
and $\b \Gamma =\begin{pmatrix}
					\log(0.1/0.9) & \log(0.1/0.9) \\
					0.5 & 0.5 \\
					1 & 1 
					\end{pmatrix}$
}\label{tab:Gacov}}
\end{table}
Scenario 2 of the Monte Carlo simulation study is based on a lower persistence of the 
Markov chain, that is obtained by letting the parameters in $\b \Gamma$ to be equal to: 
\[\b \Gamma =\begin{pmatrix}
					\log(0.4/0.6) & \log(0.4/0.6) \\
					0.5 & 0.5 \\
					1 & 1 
					\end{pmatrix}.\]

As for the basic LM model, the estimation results of $\phi_{jy|u}$ and $\b \beta$
do not differ much from those obtained in the benchmark design. We then report here only the results referred to 
the estimation of parameters in $\b \Gamma$; see Table \ref{tab:Gacov2}. In particular, 
when the Markov chain presents a lower persistence, we observe an improvement of the 
performance of the proposed estimator, and smaller differences, in terms of bias and root mean 
square error, with respect to the FML estimation method, 
especially with a large number of response variables. We also observe smaller differences between 
the two versions of the three-step estimation approach, 3S and 3S-IMP. 

\begin{table}[ht]\centering\vspace*{-2mm}
\small{
\begin{tabular}{crr|rr|rr}
 \toprule
 &\multicolumn6c{FML}\\
 $r$& \multicolumn2c{bias($\hat{\bl\Gamma}$)} &  \multicolumn2c{se($\hat{\bl\Gamma}$)} &  \multicolumn2c{rmse($\hat{\bl\Gamma}$)}\\
\midrule
\multirow3*{5}  & -0.0028 & 0.0006 & 0.1351 & 0.1287 & 0.1351 & 0.1287 \\ 
   & 0.0091 & 0.0288 & 0.1232 & 0.1147 & 0.1235 & 0.1183 \\ 
   & 0.0475 & 0.0132 & 0.1563 & 0.1512 & 0.1634 & 0.1518 \\ 
&&&&&&\\
\multirow3*{10}   & 0.0108 & -0.0119 & 0.0764 & 0.0980 & 0.0771 & 0.0987 \\ 
   & -0.0055 & 0.0047 & 0.0840 & 0.0962 & 0.0842 & 0.0963 \\ 
   & 0.0196 & 0.0272 & 0.1066 & 0.1166 & 0.1084 & 0.1198 \\
 &&&&&&\\
 \multirow3*{20}    & 0.0018 & -0.0085 & 0.0825 & 0.0809 & 0.0825 & 0.0814 \\ 
   & -0.0140 & 0.0047 & 0.0803 & 0.0717 & 0.0815 & 0.0718 \\ 
   & 0.0057 & -0.0016 & 0.0872 & 0.0850 & 0.0874 & 0.0850 \\ 
&&&&&&\\
  \multirow3*{50}  &  0.0077 & 0.0069 & 0.0655 & 0.0700 & 0.0659 & 0.0704 \\ 
   & 0.0070 & 0.0041 & 0.0653 & 0.0734 & 0.0656 & 0.0735 \\ 
  & 0.0079 & 0.0091 & 0.0740 & 0.0740 & 0.0745 & 0.0745 \\	
   \toprule
& \multicolumn6c{3S}\\
 $r$& \multicolumn2c{bias($\tilde{\bl\Gamma}$)} &  \multicolumn2c{se($\tilde{\bl\Gamma}$)} &  \multicolumn2c{rmse($\tilde{\bl\Gamma}$)}\\
\midrule  
\multirow3*{5}  &  0.2533 & 0.3936 & 0.1220 & 0.1273 & 0.2811 & 0.4137 \\ 
   & -0.3981 & -0.3984 & 0.0340 & 0.0317 & 0.3995 & 0.3997 \\ 
   & -0.7825 & -0.8156 & 0.0377 & 0.0342 & 0.7834 & 0.8163 \\ 
 &&&&&&\\ 
\multirow3*{10}   &  0.1929 & 0.2696 & 0.0648 & 0.0734 & 0.2035 & 0.2795 \\ 
   & -0.2828 & -0.2926 & 0.0483 & 0.0503 & 0.2869 & 0.2969 \\ 
   & -0.5548 & -0.5750 & 0.0438 & 0.0479 & 0.5565 & 0.5770 \\ 
&&&&&&\\
\multirow3*{20}    & 0.0884 & 0.1144 & 0.0728 & 0.0682 & 0.1145 & 0.1332 \\ 
   & -0.1319 & -0.1293 & 0.0625 & 0.0562 & 0.1459 & 0.1410 \\ 
   & -0.2446 & -0.2685 & 0.0636 & 0.0599 & 0.2528 & 0.2751 \\ 	
&&&&&&\\   
  \multirow3*{50}  & 0.0133 & 0.0142 & 0.0642 & 0.0678 & 0.0655 & 0.0693 \\ 
   & -0.0011 & -0.0047 & 0.0649 & 0.0714 & 0.0649 & 0.0715 \\ 
   & -0.0079 & -0.0071 & 0.0707 & 0.0737 & 0.0712 & 0.0740 \\ 
    \toprule
& \multicolumn6c{3S-IMP}\\
 $r$& \multicolumn2c{bias($\tilde{\bl\Gamma}$)} &  \multicolumn2c{se($\tilde{\bl\Gamma}$)} &  \multicolumn2c{rmse($\tilde{\bl\Gamma}$)}\\
\midrule   
\multirow3*{5}  & 0.1696 & 0.2701 & 0.1240 & 0.1352 & 0.2101 & 0.3020 \\ 
   & -0.2665 & -0.2640 & 0.0691 & 0.0659 & 0.2754 & 0.2721 \\ 
   & -0.5061 & -0.5663 & 0.0832 & 0.0721 & 0.5129 & 0.5709 \\ 
&&&&&&\\
\multirow3*{10}   &0.1164 & 0.1464 & 0.0715 & 0.0853 & 0.1366 & 0.1694 \\ 
   & -0.1625 & -0.1680 & 0.0690 & 0.0740 & 0.1765 & 0.1836 \\ 
   & -0.3028 & -0.3222 & 0.0709 & 0.0779 & 0.3110 & 0.3315 \\ 
 &&&&&&\\  
\multirow3*{20}    & 0.0435 & 0.0455 & 0.0787 & 0.0754 & 0.0899 & 0.0881 \\ 
   & -0.0666 & -0.0584 & 0.0741 & 0.0650 & 0.0996 & 0.0874 \\ 
   & -0.1069 & -0.1251 & 0.0788 & 0.0746 & 0.1327 & 0.1457 \\ 
&&&&&&\\ 
  \multirow3*{50}  & 0.0101 & 0.0098 & 0.0648 & 0.0692 & 0.0656 & 0.0699 \\ 
   & 0.0035 & 0.0003 & 0.0654 & 0.0723 & 0.0655 & 0.0723 \\ 
   & 0.0016 & 0.0022 & 0.0735 & 0.0740 & 0.0735 & 0.0740 \\ 
\bottomrule
\end{tabular}\vspace*{0.2cm}
\caption[Text excluding the matrix]{\em Bias, standard error (se) and root mean square error (rmse) for the full maximum likelihood and the three-step (standard and improved) estimators of $\b \Gamma$, under Scenario 2 with individual covariates and $\b \Gamma =\begin{pmatrix}
					\log(0.4/0.6) & \log(0.4/0.6) \\
					0.5 & 0.5 \\
					1 & 1 
					\end{pmatrix}$ }\label{tab:Gacov2}}
\end{table}

Under the third scenario, we evaluate the performance of the estimation algorithms with more separated
latent states. In particular, the parameters of the simulation study remain the same as in the benchmark design, apart from the 
conditional probability matrix, that is the same as in Scenario 3 of the basic LM model (Section \ref{sec:5basic}). Even in this context,
the conditional response probabilities are consistently estimated by the three-step approach. Moreover, we observe very small differences
between the results obtained by the FML estimator and the 3S and 3S-IMP estimators of parameters in $\b \beta$.
Even the 3S estimator of $\b \Gamma$ is affected by the presence of more separated latent states, with a general
improvement of its performances (Table \ref{tab:Gacov3}), especially with large $r$. 

\begin{table}[ht]\centering\vspace*{-2mm}
\small{
\begin{tabular}{crr|rr|rr}
 \toprule
  &\multicolumn6c{FML}\\
 $r$& \multicolumn2c{bias($\hat{\bl\Gamma}$)} &  \multicolumn2c{se($\hat{\bl\Gamma}$)} &  \multicolumn2c{rmse($\hat{\bl\Gamma}$)}\\
\midrule
\multirow3*{5}  &  -0.0164 & 0.0019 & 0.1355 & 0.1412 & 0.1365 & 0.1413 \\ 
 & 0.0125 & 0.0120 & 0.0852 & 0.0943 & 0.0861 & 0.0951 \\ 
   & 0.0155 & 0.0014 & 0.1026 & 0.1142 & 0.1038 & 0.1142 \\ 
&&&&&&\\ 
\multirow3*{10}   & -0.0124 & -0.0073 & 0.1120 & 0.1217 & 0.1127 & 0.1220 \\ 
   & -0.0154 & 0.0057 & 0.0811 & 0.1042 & 0.0826 & 0.1044 \\ 
   & 0.0104 & 0.0056 & 0.0979 & 0.0944 & 0.0985 & 0.0945 \\ 
&&&&&&\\ 
\multirow3*{20}    & -0.0094 & 0.0017 & 0.1227 & 0.1314 & 0.1230 & 0.1315 \\ 
   & 0.0052 & -0.0153 & 0.0818 & 0.0876 & 0.0820 & 0.0889 \\ 
   & 0.0058 & 0.0177 & 0.1036 & 0.1048 & 0.1037 & 0.1063 \\
&&&&&&\\
  \multirow3*{50}  & -0.0122 & -0.0162 & 0.1327 & 0.1323 & 0.1333 & 0.1333 \\ 
   & 0.0112 & 0.0056 & 0.0893 & 0.0980 & 0.0900 & 0.0981 \\ 
   & 0.0119 & 0.0004 & 0.1095 & 0.1070 & 0.1101 & 0.1070 \\ 
   \toprule
& \multicolumn6c{3S}\\
 $r$& \multicolumn2c{bias($\tilde{\bl\Gamma}$)} &  \multicolumn2c{se($\tilde{\bl\Gamma}$)} &  \multicolumn2c{rmse($\tilde{\bl\Gamma}$)}\\
\midrule
\multirow3*{5}  &  0.3088 & 0.3578 & 0.1066 & 0.1040 & 0.3267 & 0.3726 \\ 
  & -0.0857 & -0.0973 & 0.0679 & 0.0786 & 0.1093 & 0.1251 \\ 
   & -0.1791 & -0.2084 & 0.0856 & 0.0912 & 0.1985 & 0.2275 \\ 
&&&&&&\\ 
\multirow3*{10}   &  0.0139 & 0.0202 & 0.1096 & 0.1193 & 0.1104 & 0.1210 \\ 
   & -0.0237 & -0.0024 & 0.0797 & 0.1029 & 0.0831 & 0.1030 \\ 
   & -0.0057 & -0.0103 & 0.0950 & 0.0930 & 0.0952 & 0.0936 \\
&&&&&&\\
\multirow3*{20}    & -0.0093 & 0.0018 & 0.1227 & 0.1314 & 0.1231 & 0.1314 \\ 
   & 0.0052 & -0.0153 & 0.0818 & 0.0876 & 0.0820 & 0.0889 \\ 
  & 0.0058 & 0.0177 & 0.1036 & 0.1048 & 0.1038 & 0.1062 \\
&&&&&&\\ 
 \multirow3*{50}  & -0.0122 & -0.0162 & 0.1327 & 0.1323 & 0.1333 & 0.1333 \\ 
   & 0.0112 & 0.0056 & 0.0893 & 0.0980 & 0.0900 & 0.0981 \\ 
   & 0.0119 & 0.0004 & 0.1095 & 0.1070 & 0.1101 & 0.1070 \\
   \toprule
& \multicolumn6c{3S-IMP}\\
 $r$& \multicolumn2c{bias($\tilde{\bl\Gamma}$)} &  \multicolumn2c{se($\tilde{\bl\Gamma}$)} &  \multicolumn2c{rmse($\tilde{\bl\Gamma}$)}\\
\midrule
\multirow3*{5}  & 0.1043 & 0.1361 & 0.1253 & 0.1292 & 0.1631 & 0.1877 \\ 
 & -0.0219 & -0.0297 & 0.0776 & 0.0900 & 0.0807 & 0.0948 \\ 
   & -0.0532 & -0.0740 & 0.0964 & 0.1064 & 0.1101 & 0.1296 \\ 
&&&&&&\\   
\multirow3*{10}   &-0.0042 & 0.0004 & 0.1104 & 0.1210 & 0.1105 & 0.1210 \\ 
   & -0.0177 & 0.0037 & 0.0804 & 0.1045 & 0.0823 & 0.1046 \\ 
   & 0.0058 & 0.0017 & 0.0963 & 0.0936 & 0.0965 & 0.0936 \\ 
&&&&&&\\  
\multirow3*{20}    &-0.0093 & 0.0017 & 0.1227 & 0.1314 & 0.1230 & 0.1314 \\ 
   & 0.0052 & -0.0153 & 0.0818 & 0.0876 & 0.0820 & 0.0889 \\ 
   & 0.0058 & 0.0177 & 0.1036 & 0.1048 & 0.1037 & 0.1062 \\
&&&&&&\\ 
 \multirow3*{50}  & -0.0122 & -0.0162 & 0.1327 & 0.1323 & 0.1333 & 0.1333 \\ 
   & 0.0112 & 0.0056 & 0.0893 & 0.0980 & 0.0900 & 0.0981 \\ 
  & 0.0119 & 0.0004 & 0.1095 & 0.1070 & 0.1101 & 0.1070 \\ 
\bottomrule
\end{tabular}\vspace*{0.2cm}
\caption[Text excluding the matrix]{ \em Bias, standard error (se) and root mean square error (rmse) for the full maximum likelihood and the three-step (standard and improved) estimators of $\b \Gamma$, under Scenario 3 with individual covariates and  $\b \Gamma =\begin{pmatrix}
					\log(0.1/0.9) & \log(0.1/0.9) \\
					0.5 & 0.5 \\
					1 & 1 
					\end{pmatrix}$ }\label{tab:Gacov3}}
\end{table}

Under Scenario 4 we assess the effect of a larger number of latent states on the performances of the
estimators, by considering the following setting:
\begin{itemize}
\item number of latent states: $k=3$;
\item  conditional response probability matrix: 
$\b \Phi_j = \begin{pmatrix} 
                          0.9  & 0.1 & 0.5   \\  
                          0.1 &  0.9  & 0.5
                          \end{pmatrix} $;
\item $\b \beta = \begin{pmatrix}
					0 & 0.5 & 1 & 0 & 0.5 & 1
				\end{pmatrix}\tr$;         
\item $\b \Gamma =\begin{pmatrix}
					\log(0.4/0.6) & \log(0.4/0.6)& \log(0.4/0.6) \\
					0.5 & 0.5 & 0.5 \\
					1 & 1 & 1 \\
					\log(0.4/0.6) & \log(0.4/0.6)& \log(0.4/0.6) \\
					0.5 & 0.5 & 0.5 \\
					1 & 1 & 1 
					\end{pmatrix}$. 
\end{itemize}
The remaining parameters are left unchanged with respect to Scenario 1.
Even under this scenario, we can get the same conclusion about estimation of $\phi_{jy|u}$ and $\b \beta$ as 
in the previous scenarios. About estimation of parameters in $\b \Gamma$, from Table \ref{tab:Gacov4}, in which 
we report only the results for $r=5$ and $r=50$ for reason of space, we observe that the bias of the proposed estimator is quite high for 
$r=5$ but it becomes negligible for larger $r$, even in the standard version (3S) of the estimator.
The behavior of the 3S and 3S-IMP methods, in terms of standard error and root mean square error, is comparable
 with those of the FML approach. 

\begin{table}[ht]\centering\vspace*{0.25cm}
\small{
\begin{tabular}{rrrr|rrr|rrr}
 \toprule
 \multicolumn{10}c{FML}\\
 $r$& \multicolumn3c{bias($\hat{\bl\Gamma}$)} &  \multicolumn3c{se($\hat{\bl\Gamma}$)} &  \multicolumn3c{rmse($\hat{\bl\Gamma}$)}\\
\midrule
\multirow6*{5}  &  -0.025 & -0.027 & -0.020 & 0.155 & 0.128 & 0.137 & 0.157 & 0.131 & 0.139 \\ 
 &0.011 & 0.010 & 0.013 & 0.127 & 0.125 & 0.123 & 0.127 & 0.125 & 0.123 \\ 
  &-0.008 & 0.009 & -0.005 & 0.155 & 0.154 & 0.134 & 0.155 & 0.155 & 0.135 \\ 
  & -0.032 & 0.011 & -0.014 & 0.141 & 0.161 & 0.158 & 0.145 & 0.161 & 0.158 \\
  &-0.003 & -0.012 & 0.003 & 0.117 & 0.109 & 0.135 & 0.117 & 0.110 & 0.135 \\ 
   & -0.007 & 0.026 & -0.003 & 0.126 & 0.142 & 0.132 & 0.126 & 0.144 & 0.132 \\
&&&&&&&&&\\
\multirow6*{50}   & 0.006 & 0.001 & 0.003 & 0.085 & 0.086 & 0.070 & 0.085 & 0.086 & 0.070 \\ 
   & -0.002 & 0.005 & -0.005 & 0.071 & 0.068 & 0.068 & 0.071 & 0.068 & 0.069 \\ 
   & 0.010 & -0.006 & -0.001 & 0.081 & 0.067 & 0.071 & 0.081 & 0.067 & 0.071 \\ 
   & 0.002 & 0.002 & 0.009 & 0.069 & 0.073 & 0.068 & 0.069 & 0.073 & 0.069 \\ 
   & 0.004 & 0.010 & 0.007 & 0.060 & 0.071 & 0.067 & 0.060 & 0.071 & 0.068 \\ 
   & 0.005 & 0.004 & 0.010 & 0.082 & 0.082 & 0.080 & 0.082 & 0.082 & 0.081 \\ 
\toprule
 \multicolumn{10}c{3S}\\
 $r$& \multicolumn3c{bias($\tilde{\bl\Gamma}$)} &  \multicolumn3c{se($\tilde{\bl\Gamma}$)} &  \multicolumn3c{rmse($\tilde{\bl\Gamma}$)}\\ 
\midrule
\multirow6*{5}  & 0.256 & 0.341 & 0.080 & 0.129 & 0.112 & 0.158 & 0.287 & 0.359 & 0.178 \\ 
   & -0.330 & -0.444 & -0.213 & 0.036 & 0.038 & 0.054 & 0.332 & 0.446 & 0.219 \\ 
   & -0.667 & -0.891 & -0.438 & 0.039 & 0.033 & 0.056 & 0.669 & 0.891 & 0.442 \\ 
   & 0.041 & 0.304 & 0.291 & 0.173 & 0.121 & 0.141 & 0.178 & 0.327 & 0.323 \\ 
   & -0.202 & -0.433 & -0.341 & 0.046 & 0.035 & 0.041 & 0.208 & 0.434 & 0.343 \\ 
   & -0.395 & -0.863 & -0.682 & 0.052 & 0.036 & 0.043 & 0.398 & 0.864 & 0.684 \\ 
&&&&&&&&&\\
\multirow6*{50} & 0.000 & 0.008 & 0.006 & 0.071 & 0.076 & 0.083 & 0.071 & 0.077 & 0.084 \\ 
   & -0.004 & 0.000 & 0.011 & 0.069 & 0.063 & 0.063 & 0.069 & 0.063 & 0.064 \\ 
   & 0.002 & -0.018 & 0.013 & 0.078 & 0.068 & 0.076 & 0.078 & 0.070 & 0.077 \\ 
   & -0.002 & 0.004 & 0.014 & 0.065 & 0.079 & 0.077 & 0.065 & 0.079 & 0.078 \\ 
   & -0.001 & -0.006 & 0.008 & 0.067 & 0.071 & 0.071 & 0.067 & 0.072 & 0.071 \\ 
   & 0.001 & -0.014 & 0.014 & 0.086 & 0.073 & 0.077 & 0.086 & 0.075 & 0.078 \\ 
\toprule
 \multicolumn{10}c{3S-IMP}\\
 $r$& \multicolumn3c{bias($\tilde{\bl\Gamma}$)} &  \multicolumn3c{se($\tilde{\bl\Gamma}$)} &  \multicolumn3c{rmse($\tilde{\bl\Gamma}$)}\\ 
\midrule   
\multirow6*{5} & 0.186 & 0.192 & -0.003 & 0.139 & 0.122 & 0.167 & 0.233 & 0.227 & 0.167 \\ 
   & -0.193 & -0.330 & -0.118 & 0.082 & 0.072 & 0.075 & 0.209 & 0.338 & 0.140 \\ 
   & -0.417 & -0.663 & -0.251 & 0.088 & 0.066 & 0.085 & 0.426 & 0.666 & 0.265 \\ 
   & -0.002 & 0.189 & 0.223 & 0.180 & 0.128 & 0.159 & 0.180 & 0.228 & 0.274 \\ 
   & -0.104 & -0.327 & -0.230 & 0.065 & 0.066 & 0.086 & 0.123 & 0.334 & 0.245 \\ 
   & -0.194 & -0.651 & -0.450 & 0.076 & 0.071 & 0.097 & 0.209 & 0.655 & 0.460 \\ 
&&&&&&&&&\\
  \multirow6*{50} & -0.000 & 0.006 & 0.006 & 0.071 & 0.076 & 0.083 & 0.071 & 0.076 & 0.084 \\ 
   & -0.003 & 0.002 & 0.012 & 0.069 & 0.063 & 0.062 & 0.069 & 0.063 & 0.064 \\ 
   & 0.004 & -0.014 & 0.014 & 0.078 & 0.068 & 0.076 & 0.078 & 0.070 & 0.078 \\ 
   & -0.002 & 0.003 & 0.013 & 0.065 & 0.079 & 0.077 & 0.065 & 0.079 & 0.078 \\ 
   & -0.000 & -0.005 & 0.009 & 0.067 & 0.072 & 0.071 & 0.067 & 0.072 & 0.072 \\ 
   & 0.002 & -0.010 & 0.016 & 0.086 & 0.074 & 0.077 & 0.086 & 0.074 & 0.078 \\
\bottomrule
\end{tabular}\vspace*{0.2cm}
\caption[Text excluding the matrix]{ \em Bias, standard error (se) and root mean square error (rmse) for the full maximum likelihood and the three-step (standard and improved) estimators of $\b \Gamma$, under Scenario 4 with individual covariates }\label{tab:Gacov4}}
\end{table}

As demonstrated in the context of LC models by \cite{bolck:et:al:04} and, more recently, by \cite{vermunt:2010} and \cite{bakk:et:al:13},
all the results above show a bias in the parameter estimates due to the classification error introduced in the second step, which confirms that, in general, the three-step approach underestimates the relationship between covariates and class membership. We note this  bias in the estimates of the parameters affecting both the initial and transition probabilities.  
As already stated, this bias becomes negligible when the number of response variables increases; moreover, the proposed 3S-IMP version allows us to obtain an improvement of the behavior of the corresponding estimator, with a 
less evident  bias.

Finally, the inclusion of individual covariates confirms the advantages, in term of computational cost, of the 
three-step estimation method with respect to the full likelihood method. In all scenarios, 
the computing time required by both 3S and 3S-IMP algorithms is significantly lower than
that required by the FML algorithm. 
The number of  iterations required by the 3S-IMP version for convergence goes from a maximum of 37 (with $r=5$ in Scenario 4) to a 
minimum of 1 (with $r=50$ in Scenario 3).


\section{Empirical illustration}
\label{sec:ulisse}
In order to illustrate the proposed estimation approach, we
outline an application based on a real dataset which is derived from a project, named ULISSE
(``Un Link Informatico sui Servizi Sanitari
Esistenti per l'Anziano'' - ``A Computerized Network on
Health Care Services for Older People'') aimed at studying the
health status of elderly patients who currently receive health care
assistance in Italy; see \cite{lattanzio2010}.  In the analysis here
presented, we consider only data referred to the health condition of
patients hosted in 26 nursing homes, which cover $n=911$ patients. 
The project is based on a longitudinal survey in which the patients were evaluated at admission and then
re-evaluated at 6 and 12 month after the admission; therefore we consider $T=3$ time occasions. 

The original questionnaire is made of $r=75$ polytomously-responded
items, with categories ordered according to increasing
difficulty levels in accomplishing a certain task or severeness of a
specific aspect of the health conditions. These items are grouped
into eight different sections of the questionnaire, concerning:
\begin{enumerate}
\item Cognitive Conditions (CC);
\item Auditory and View Fields (AVF); 
\item Humor and Behavioral Disorders (HBD);
\item Activities of Daily Living (ADL);
\item Incontinence (I); 
\item Nutritional Field (NF);
\item Dental Disorder (DD);
\item Skin Conditions (SC).
\end{enumerate}

The complete list of items for each section, with the corresponding number of responses categories, is reported in Appendix. 
In this application, the missing responses to a given item are dealt with the 
{\em missing at random} assumption \citep{rubin:76,litl:rubin:1987}.

The extended LM model allows us to consider time-constant and time-varying individual covariates. In particular,
among time-constant covariates we include gender, and dummy variables for coding
the nursing home to which the subject belongs, whereas among time-varying covariates we include
age and time interval between occasions of administration of the questionnaire. Accordingly, we have
$p_1=27$ covariates affecting the logit for the initial probabilities and $p_2=28$ covariates
affecting the logit for the transition probabilities. 
In order to make the model more parsimonious, 
we rely on the following parameterization for the transition probabilities, 
that is, a multinomial logit based on the difference between two sets of parameters
\[
\log \frac{p(U^{(t)}=v|U^{(t-1)}=u,\b x^{(t)})}
{p(U^{(t)}=u|U^{(t-1)}=u,\b x^{(t)})} =
\gamma_{0uv}+(\b x_i^{(t)})\tr (\b\gamma_{1u}-\b\gamma_{1v}).
\quad t \geq 2, \;\;\;u\neq v,
\]
where $\b\gamma_{11}=0$ to ensure model identifiability.
The parameterization used for modeling the initial probabilities is again based on standard multinomial logit, as defined in (\ref{eq:be}).
For this illustration we consider a fixed number of latent classes, $k=4$.
Even considering the above parametrization, the number of free parameters to be estimated
is very large and the computing time required by the full likelihood approach may become 
excessive, especially because the EM algorithm requires a large number of random initializations to 
increase the chance of the convergence to the global maximum of the model log-likelihood. In such a context, 
the computational cost of the three-step algorithm is significantly lower than that of the
full likelihood approach, while reaching very similar performance in terms of parameter estimates.

In applying the three-step estimation method to this data, we first fit the
basic LC model for the set of response variables, in which the responses provided
by the same subject at different occasions are considered as corresponding 
to separate sample units. By using this approach, we can initialize the EM algorithm by a number
 of random starting values equal to $100(k -1)$ and take the estimates corresponding to the highest log-likelihood 
 at convergence of the algorithm.  

From the first step of the proposed approach we obtain the final estimates of the conditional response probabilities. Since the items are categorical, with a different number of
categories, we compute the  following {\em item mean score} to make the interpretation of the results easier:
\[\hat{\mu}_{j|u} = \frac{1}{c_j-1}\sum_y
y\;\hat{\phi}_{jy|u},\quad j=1,\ldots,r,\;\;\; u=1,\ldots,k, \;\;\; y=0,\ldots,c_j-1.  
\]
In particular, a value of $\hat{\mu}_{j|u}$ close to 0 corresponds
to a low probability of suffering from a certain pathology, whereas
a value close to 1 corresponds to a high probability of suffering
from the same pathology. To summarize these results, we also compute the  {\em section mean score}
$\hat{\bar{\mu}}_{d|u}$ as the average of $\hat{\mu}_{j|u}$ for the
items composing each section $d$ of the
questionnaire, with $d=1,\ldots,8$.

In order to interpret the results
we order the latent states on the basis of the values of
$\hat{\bar{\mu}}_{d|u}$ assumed in the section denoted by ADL
(Activity of Daily Living) of the questionnaire, which is the section with the
highest difference between the maximum and the minimum value of the
section mean score across states.
For each latent state, Table~\ref{tab:average} shows the values of
$\hat{\bar{\mu}}_{d|u}$, together with the difference between the maximum and the minimum value
of $\hat{\bar{\mu}}_{d|u}$ for each section of the questionnaire. 

\begin{table}[ht]\centering\vspace*{0.25cm}
\small{
\begin{tabular}{p{1.8cm}|rrrrrrrr}
\toprule & \multicolumn8{c}{$d$}\\
 \cline{2-9} & \multicolumn1c{1}
& \multicolumn1c{2} & \multicolumn1c{3} &
\multicolumn1c{4} & \multicolumn1c{5}& \multicolumn1c{6} & \multicolumn1c{7}& \multicolumn1{c}{8}    \\
 \multicolumn1{c|}{$u$} & \multicolumn1c{(CC)}
& \multicolumn1c{(AVF)} & \multicolumn1c{(HBD)} &
\multicolumn1c{(ADL)} & \multicolumn1c{(I)}& \multicolumn1c{(NF)} &\multicolumn1c{(DD)}&  \multicolumn1{c}{(SC)} \\
\midrule
\multicolumn1{c|}{1}	&	0.1083 & 0.1356 & 0.0943 & 0.1148 & 0.2837 & 0.0602 & 0.2214 & 0.0227 \\ 
\multicolumn1{c|}{2}	& 0.6227 & 0.4023 & 0.2335 & 0.2912 & 0.6653 & 0.0864 & 0.2133 & 0.0204 \\
\multicolumn1{c|}{3}	& 0.1927 & 0.1896 & 0.1246 & 0.6018 & 0.6736 & 0.0929 & 0.2246 & 0.0546 \\ 
\multicolumn1{c|}{4}	& 0.7063 & 0.5718 & 0.1436 & 0.7850 & 0.8977 & 0.1566 & 0.2232 & 0.1076 \\ 
\midrule
$\max_{u}(\hat{\bar{\mu}}_{d|u}) - \min_u(\hat{\bar{\mu}}_{d|u})$     &   0.5980  &  0.4362 & 0.1392 & 0.6702 & 0.6140   & 0.0964 & 0.0113 & 0.0872      \\
 \bottomrule
\end{tabular}
\vspace*{0.2cm}
\caption{\em Estimated section mean score, $\hat{\bar{\mu}}_{d|u}$,
for each latent state $u$ and each section $d$ of the questionnaire
together with the difference
between the largest and the smallest estimated section mean score
for each section.}\label{tab:average}}
\end{table}

As we can see, smaller differences are observed for sections
DD, SC, NF, and HBD which, consequently, tend to discriminate less between subjects with respect to the other sections.
The first state corresponds to the best health conditions with respect to all the pathologies measured
by the sections of the questionnaire, apart from DD and SC. 
On the other hand, the fourth state corresponds to cases with 
the worst health conditions for almost all the pathologies. Intermediate states show a different case-mix
depending on the section mean score pattern. In particular, the second state corresponds to patients with poor cognitive and auditory 
and view conditions (CC and AVF).  The third state corresponds to patients with deteriorated physical  
conditions with respect to the pathologies measured by section ADL. 
 
Once the conditional response probabilities are estimated
by the first step of the proposed approach, we can 
estimate the parameters $\b \beta = (\beta_{0u},\b\beta_{1u}\tr)\tr$, $\gamma_{0uv}$ and $\b \gamma_{1u}$ by fitting multinomial logit models with weights.
In Table \ref{tab:beGa} we report the results of this step. The same table also shows the significance of the estimates. The standard errors of the estimates has been obtained by a non-parametric bootstrap based on 99 samples. 
Note that, in this application, a bootstrap resampling procedure in which the parameters are estimated by means of the full likelihood approach
is unfeasible, due to the considerable computational cost required by this approach.

\begin{table}[ht]\centering\vspace*{0.25cm}
\small{
\begin{tabular}{rrrrrrr}
  \toprule
  & \multicolumn3c{initial probabilities} & \multicolumn3c{transition probabilities} \\
 & \multicolumn1c{$\b\beta_2$} &  \multicolumn1c{$\b\beta_3$} &  \multicolumn1c{$\b\beta_4$} & \multicolumn1c{$\b\gamma_2$} & \multicolumn1c{$\b\gamma_3$} & \multicolumn1c{$\b\gamma_4$}\\ 
 \midrule
intercept	&	0.0156	\hspace{4mm}	&	-1.6610	\hspace{4mm}	&	-3.4234	$^{**}$\hspace{1.2mm} 	&	 \multicolumn1c{-} 		&	 \multicolumn1c{-} 		&	 \multicolumn1c{-} 		\\
gender 	&	0.3756	\hspace{4mm}	&	0.5485	$^*$\hspace{2.5mm} 	&	0.7037	$^{**}$\hspace{1.2mm} 	&	-0.0532	\hspace{4mm}	&	-0.1858	\hspace{4mm}	&	-0.0261	\hspace{4mm}	\\
nh$_1$	&	-0.1961	\hspace{4mm}	&	0.8121	\hspace{4mm}	&	0.0020	\hspace{4mm}	&	-1.0676	\hspace{4mm}	&	-0.0694	\hspace{4mm}	&	-0.9278	\hspace{4mm}	\\
nh$_2$	&	-1.9504	$^{**}$\hspace{1.2mm} 	&	-0.3918	\hspace{4mm}	&	-2.3623	$^{***}$	&	-2.2406	$^*$\hspace{2.5mm} 	&	-1.3643	\hspace{4mm}	&	-2.5257	$^*$\hspace{2.5mm} 	\\
nh$_3$	&	-0.1054	\hspace{4mm}	&	-0.4334	\hspace{4mm}	&	-0.2467	\hspace{4mm}	&	0.9412	\hspace{4mm}	&	1.1634	\hspace{4mm}	&	0.4836	\hspace{4mm}	\\
nh$_4$	&	-1.0478	$^{\displaystyle{\cdot}}$\hspace{3mm} 	&	-0.8929	\hspace{4mm}	&	-1.4349	$^{**}$\hspace{1.2mm} 	&	-1.0167	\hspace{4mm}	&	-0.9319	\hspace{4mm}	&	-1.1176	\hspace{4mm}	\\
nh$_5$	&	-1.2677	$^{\displaystyle{\cdot}}$\hspace{3mm} 	&	-1.0193	\hspace{4mm}	&	-0.5042	\hspace{4mm}	&	-1.6953	$^{\displaystyle{\cdot}}$\hspace{3mm} 	&	-0.4376	\hspace{4mm}	&	-1.5519	\hspace{4mm}	\\
nh$_6$	&	-0.6045	\hspace{4mm}	&	-1.1527	\hspace{4mm}	&	-0.5333	\hspace{4mm}	&	-1.1935	\hspace{4mm}	&	-0.0120	\hspace{4mm}	&	-1.8624	\hspace{4mm}	\\
nh$_7$	&	-3.0934	\hspace{4mm}	&	-0.7201	\hspace{4mm}	&	-0.3363	\hspace{4mm}	&	-1.6755	\hspace{4mm}	&	-0.5311	\hspace{4mm}	&	-1.4346	\hspace{4mm}	\\
nh$_8$	&	0.3198	\hspace{4mm}	&	0.8767	\hspace{4mm}	&	0.8430	\hspace{4mm}	&	-1.7931	\hspace{4mm}	&	-1.0027	\hspace{4mm}	&	-2.4452	$^*$\hspace{2.5mm} 	\\
nh$_9$	&	0.0048	\hspace{4mm}	&	0.2805	\hspace{4mm}	&	-0.7442	\hspace{4mm}	&	-2.2255	$^*$\hspace{2.5mm} 	&	-1.4504	$^{\displaystyle{\cdot}}$\hspace{3mm} 	&	-2.5039	$^{**}$\hspace{1.2mm} 	\\
nh$_{10}$	&	-2.0851	\hspace{4mm}	&	-1.4617	\hspace{4mm}	&	-3.5955	\hspace{4mm}	&	-2.9837	\hspace{4mm}	&	-1.1377	\hspace{4mm}	&	-2.9800	$^*$\hspace{2.5mm} 	\\
nh$_{11}$	&	-1.6932	$^*$\hspace{2.5mm} 	&	-1.0069	\hspace{4mm}	&	-1.8672	$^{**}$\hspace{1.2mm} 	&	-2.8173	$^*$\hspace{2.5mm} 	&	-0.9797	\hspace{4mm}	&	-2.5322	$^*$\hspace{2.5mm} 	\\
nh$_{12}$	&	0.7904	\hspace{4mm}	&	0.5132	\hspace{4mm}	&	1.0443	\hspace{4mm}	&	-1.2102	\hspace{4mm}	&	-1.5142	\hspace{4mm}	&	-2.0071	\hspace{4mm}	\\
nh$_{13}$	&	-1.5811	$^{\displaystyle{\cdot}}$\hspace{3mm} 	&	-2.3327	\hspace{4mm}	&	-1.6188	$^{**}$\hspace{1.2mm} 	&	-2.3505	$^{**}$\hspace{1.2mm} 	&	-0.9972	\hspace{4mm}	&	-2.2378	$^{**}$\hspace{1.2mm} 	\\
nh$_{14}$	&	-0.5118	\hspace{4mm}	&	-0.4459	\hspace{4mm}	&	-6.3197	$^{***}$	&	-2.5696	$^*$\hspace{2.5mm} 	&	-0.3511	\hspace{4mm}	&	-1.9662	\hspace{4mm}	\\
nh$_{15}$	&	-2.1905	\hspace{4mm}	&	-0.3920	\hspace{4mm}	&	-1.4533	$^{\displaystyle{\cdot}}$\hspace{3mm} 	&	-1.5935	\hspace{4mm}	&	-1.4392	\hspace{4mm}	&	-2.1450	\hspace{4mm}	\\
nh$_{16}$	&	0.3121	\hspace{4mm}	&	0.2009	\hspace{4mm}	&	0.0735	\hspace{4mm}	&	-0.2615	\hspace{4mm}	&	0.6656	\hspace{4mm}	&	-0.2236	\hspace{4mm}	\\
nh$_{17}$	&	-0.6438	\hspace{4mm}	&	-0.0830	\hspace{4mm}	&	-0.2925	\hspace{4mm}	&	-1.4858	$^{\displaystyle{\cdot}}$\hspace{3mm} 	&	-0.9442	\hspace{4mm}	&	-2.7412	$^{***}$	\\
nh$_{18}$	&	5.5294	$^{***}$	&	-1.1376	\hspace{4mm}	&	5.5737	$^{***}$	&	0.9654	\hspace{4mm}	&	1.7519	\hspace{4mm}	&	-0.3829	\hspace{4mm}	\\
nh$_{19}$	&	-0.1291	\hspace{4mm}	&	0.5499	\hspace{4mm}	&	0.2007	\hspace{4mm}	&	-1.5408	$^{\displaystyle{\cdot}}$\hspace{3mm} 	&	-1.6047	\hspace{4mm}	&	-1.8218	$^*$\hspace{2.5mm} 	\\
nh$_{21}$	&	0.1693	\hspace{4mm}	&	-0.8810	\hspace{4mm}	&	-0.6420	\hspace{4mm}	&	-1.0734	\hspace{4mm}	&	-1.6393	\hspace{4mm}	&	-1.8884	\hspace{4mm}	\\
nh$_{22}$	&	-0.1079	\hspace{4mm}	&	0.4034	\hspace{4mm}	&	-0.8618	\hspace{4mm}	&	-0.5954	\hspace{4mm}	&	-0.7328	\hspace{4mm}	&	-2.4313	\hspace{4mm}	\\
nh$_{23}$	&	0.9074	\hspace{4mm}	&	2.0485	\hspace{4mm}	&	2.4644	\hspace{4mm}	&	-2.5099	\hspace{4mm}	&	-2.2295	\hspace{4mm}	&	-3.2658	\hspace{4mm}	\\
nh$_{24}$	&	-1.8094	$^*$\hspace{2.5mm} 	&	-0.2122	\hspace{4mm}	&	-0.8197	\hspace{4mm}	&	-2.1891	$^*$\hspace{2.5mm} 	&	-0.8078	\hspace{4mm}	&	-2.3811	$^{**}$\hspace{1.2mm} 	\\
nh$_{25}$	&	-0.0354	\hspace{4mm}	&	0.9532	\hspace{4mm}	&	-0.2216	\hspace{4mm}	&	-1.9191	$^{\displaystyle{\cdot}}$\hspace{3mm} 	&	-1.7159	\hspace{4mm}	&	-3.3229	$^*$\hspace{2.5mm} 	\\
nh$_{26}$	&	-0.0430	\hspace{4mm}	&	-0.1403	\hspace{4mm}	&	-0.4727	\hspace{4mm}	&	-1.9055	$^*$\hspace{2.5mm} 	&	-0.6472	\hspace{4mm}	&	-2.6104	$^{**}$\hspace{1.2mm} 	\\
diff-time	&	 \multicolumn1c{-} 		&	 \multicolumn1c{-} 		&	 \multicolumn1c{-} 		&	0.1753	\hspace{4mm}	&	-0.4054	\hspace{4mm}	&	-0.0320	\hspace{4mm}	\\
age	&	0.0003	\hspace{4mm}	&	0.0149	\hspace{4mm}	&	0.0427	$^{**}$\hspace{1.2mm} 	&	0.0370	$^{\displaystyle{\cdot}}$\hspace{3mm} 	&	0.0542	$^*$\hspace{2.5mm} 	&	0.0766	$^{**}$\hspace{1.2mm} 	\\
\bottomrule
\end{tabular}}
\begin{center}
\scriptsize{$^{***}$significant at the 0.1\% level - $^{**}$significant at the 1\% level - $^*$significant at the 5\% level - $^{\displaystyle{\cdot}}$significant at the 10\% level}
\end{center}
\caption{\em Estimates of the regression parameter affecting the distribution of the initial and transition probabilities.}\label{tab:beGa}
\end{table}

From the table we note that gender has a significant effect on the logit of the initial probability, especially
on the probability of being in the last state with respect to the first. The same can be said for age, which also shows 
a significant effect on the transition probabilities. Moreover, we note that certain nursing homes have
a higher effect on both the initial and transition probabilities with respect to others. In such a context, it may be
of interest to study the ability of the facilities to retain over time patients in the initial latent state.

For an easier interpretation of the results, we also report the estimated initial and transition probabilities of the latent
Markov process.
More in detail, Table \ref{tab:init} shows the means of the estimated initial probabilities, $\bar{\tilde\pi}_u$, $u=1,\ldots,k$, with $k=4$, for age classes and gender. The results show that, at the beginning 
of the study, males have a higher probability of being in the first state,
 which corresponds to the best health condition, with respect to females. This is reasonable,
 as the health status of women who require admission to nursing homes is likely worse than that of men of the same age.
This because women are typically  
more able to take care of themselves. 
Moreover, it is obvious that older patients (both male and female) have a higher probability
of being in in fifth state, corresponding to the worse health condition.  

\begin{table}[ht]\centering\vspace*{0.25cm}
\small{
\begin{tabular}{lcrrr}
\toprule
    &     laten states  &  \multicolumn2{c}{gender}          \\
age   &   $u$  &   \multicolumn1{c}{M}   &   \multicolumn1{c}{F} &  overall \\
\midrule
$\leq 75$    &  1 & 0.4419 & 0.2932 & 0.3672 \\ 
&  2 & 0.2402 & 0.2608 & 0.2506 \\ 
 & 3 & 0.1584 & 0.2017 & 0.1801 \\ 
  &4 & 0.1595 & 0.2442 & 0.2021 \\ 
    &       &       &       \\
$75< \rm{age} \leq 85$ &   1   & 0.3780 & 0.2819 & 0.3106 \\ 
  &2 & 0.2161 & 0.1982 & 0.2036 \\ 
  &3 & 0.1817 & 0.2199 & 0.2085 \\ 
  &4 & 0.2242 & 0.3000 & 0.2773 \\ 	
     &       &       &       \\
$> 85$ 	&   1   & 0.3626 & 0.2610 & 0.2760 \\ 
  & 2 & 0.1777 & 0.1549 & 0.1583 \\ 
  & 3 & 0.1757 & 0.2236 & 0.2165 \\ 
  & 4 & 0.2840 & 0.3605 & 0.3493 \\ 
  \bottomrule
\end{tabular}
\vspace*{0.2cm}
\caption{\em Means of the estimated initial probabilities,  $\bar{\tilde\pi}_u$,
over all the different nursing homes, for age classes and
gender}\label{tab:init}}
\end{table}

In Table \ref{tab:trans_all} we report the overall means of the estimated transition probabilities, $\tilde{\pi}^{(t)}_{v|u}$, for all subjects in the sample, whereas in Table \ref{tab:trans_nh} we report the means of $\tilde{\pi}^{(t)}_{v|u}$ computed for two nursing homes with a similar number of patients but with different estimated transition matrices.     

\begin{table}[ht]\centering\vspace*{0.25cm}
\small{
\begin{tabular}{rrrrrr}
\toprule
       &       & \multicolumn4{c}{latent state $u$}\\
  $t$   & latent state  $v$   &    \multicolumn1{c}{1}   &    \multicolumn1{c}{2}   &    \multicolumn1{c}{3} &    \multicolumn1{c}{4}    \\
\midrule
   2   & 1 & 0.8417 & 0.0728 & 0.0703 & 0.0151 \\ 
  &2 & 0.0259 & 0.7781 & 0.0474 & 0.1486 \\ 
 & 3 & 0.0278 & 0.0525 & 0.8382 & 0.0815 \\ 
 & 4 & 0.0000 & 0.0357 & 0.0217 & 0.9426 \\ 
 \vspace{-1.5mm}
\\
   3   &1 & 0.8377 & 0.0734 & 0.0733 & 0.0156 \\ 
 & 2 & 0.0253 & 0.7747 & 0.0495 & 0.1505 \\ 
  &3 & 0.0272 & 0.0531 & 0.8359 & 0.0838 \\ 
  &4 & 0.0000 & 0.0349 & 0.0222 & 0.9429 \\ 
\bottomrule
\end{tabular}
\vspace*{0.2cm}
\caption{\em Overall means of the estimated transition
probabilities, $\bar{\tilde{\pi}}^{(t)}_{v|u}$ }\label{tab:trans_all}}
\end{table}

\begin{table}[ht]\centering\vspace*{0.25cm}
\small{
\begin{tabular}{rrrrrrr}
\toprule
 &      &       & \multicolumn4{c}{latent state $u$}\\
 & $t$   & latent state  $v$   &    \multicolumn1{c}{1}   &    \multicolumn1{c}{2}   &    \multicolumn1{c}{3} &    \multicolumn1{c}{4}    \\
  \midrule
\multirow9*{nh$_{3}$} &  2   & 1 & 0.4320 & 0.3129 & 0.2057 & 0.0494 \\ 
 && 2 & 0.0019 & 0.8429 & 0.0273 & 0.1279 \\ 
 && 3 & 0.0031 & 0.0704 & 0.8268 & 0.0997 \\ 
 && 4 & 0.0000 & 0.0354 & 0.0135 & 0.9511 \\ 
 \vspace{-1.5mm}
\\
&   3   &1 & 0.4256 & 0.3057 & 0.2184 & 0.0503 \\ 
&&  2 & 0.0019 & 0.8360 & 0.0292 & 0.1329 \\ 
 && 3 & 0.0030 & 0.0671 & 0.8332 & 0.0966 \\ 
  &&4 & 0.0000 & 0.0344 & 0.0138 & 0.9518 \\ 
\midrule
\multirow9*{nh$_{25}$} &  2   &1 & 0.9359 & 0.0366 & 0.0253 & 0.0022 \\ 
 && 2 & 0.0372 & 0.8841 & 0.0294 & 0.0493 \\ 
 && 3 & 0.0581 & 0.0695 & 0.8366 & 0.0357 \\ 
 && 4 & 0.0001 & 0.0885 & 0.0345 & 0.8769 \\ 
 \vspace{-1.5mm}
\\
&   3   &1 & 0.9354 & 0.0405 & 0.0219 & 0.0023 \\ 
 && 2 & 0.0342 & 0.8957 & 0.0234 & 0.0467 \\ 
&&  3 & 0.0650 & 0.0859 & 0.8079 & 0.0413 \\ 
 && 4 & 0.0001 & 0.0947 & 0.0289 & 0.8763 \\
\bottomrule
\end{tabular}
\vspace*{0.2cm}
\caption{\em Overall means of the estimated transition
probabilities, $\bar{\tilde{\pi}}^{(t)}_{v|u}$ for two different nursing homes}\label{tab:trans_nh}}
\end{table}

According to the tables, we observe a quite high persistence in the same state, especially for the first 
and the last state. Moreover, patients classified in the second state present an average transition probability 
toward the fourth state, corresponding to severe health conditions, of around 0.15.
We also note a very low probability of transition
from higher to lower states, meaning that it is quite unlikely an improvement of the patients health conditions. 
This probability drops to zero when we consider the transition from the fourth state to the first state, corresponding to 
the best health status. We finally observe that different nursing homes may have a very different behavior. 
As an example, nursing home nh$_{3}$ presents a low persistence in the first latent state with a quite 
high probability of transition toward the second and the third latent states (between 0.2 and 0.3). On the other hand,
nursing home nh$_{25}$ presents very high persistence probabilities. Moreover, the transition probabilities toward
latent states corresponding to deteriorated conditions are alway lower than 0.05.

\section{Conclusions and further developments}\label{sec:6}

In this paper, we propose an extended version of the three-step approach \citep{vermunt:2010}  
to estimate LM models, which are typically used in 
applications involving longitudinal data. In such a context, it is not uncommon that
the number of response variables for each time occasion is very large, and many
latent states are specified. Moreover, time-constant and time-varying individual 
covariates may be included in the model. 

The proposed estimation method may represent a valid alternative to the full maximum likelihood 
approach, which is typically based on the EM algorithm. 
More in detail, when applied to LM models, the proposed three-step estimation consists
of a preliminary clustering of the subjects, on the basis of the time specific responses only,
in which the responses provided by the same subject to different time occasions are considered as 
separated units. Moreover, every sample unit is not strictly assigned to a latent state
at each occasion, and this state may change across time, so as to allow the estimation of the parameters
of the latent Markov process. We also propose an improved version of the three-step estimator, in which the last two steps
of the algorithm are iterated until convergence, while keeping fixed the results from the first step. 

We provide a proof that this approach leads to consistent estimates of the
conditional response probabilities. Moreover, when the number of response variables tends to infinity, 
even the parameters of the latent process are consistently estimated. 

We perform a simulation study aimed at assessing the behavior of the proposed 
estimation algorithms under different scenarios, in comparison with that of the full maximum
likelihood approach. On the basis of the results of this study, it is possible to conclude that the behavior,
in terms of bias and root mean square error, of the 
three-step estimator improves with the number of response variables, 
with the separation between latent states, and 
as the level of persistence of the latent Markov chain decreases.  
Moreover, we observe that, under certain scenarios, the improved version of the three-step method outperforms
the standard one, even significantly, with a reduction of the bias, already detected in the
context of  LC model, in the estimates 
of the relationship between class membership and covariates  \citep{bolck:et:al:04,vermunt:2010,bakk:et:al:13}.
  
The potential of the proposed approach is also illustrated 
with an application to real data, involving
a very large number of response variables and individual covariates. 
In such a complex study, we observe that the proposed three-step
estimation method may overcome some
of the typical drawbacks of the full likelihood estimation, first of all the slowness to converge and 
the presence of many local maxima. Moreover, a bootstrap resampling procedure 
aimed at computing the standard errors of the estimates may be performed with
a reasonable effort.  The advantage in terms of computational cost may also be 
exploited to perform cross-validation, so as to select the number of latent states \citep{smyth:00}.
Finally, the results of the three-step estimation method may be used to define sensible starting values of the EM algorithm 
in the full maximum likelihood approach,
so as to prevent the problem of the multimodality of the model log-likelihood. 
All these advantages are obtained while having very similar results
in terms of parameter estimates.  

Finally, It is worth noting that, throughout the paper, we
consider the case of categorical response variables because this is
the typical case of application of LM models. However, we plan to extend the proposed estimation 
approach when the outcomes are all
continuous, or in the 
case of mixed outcomes, and of latent variables with different state
spaces.  Moreover, further research is necessary to find a method for easily obtaining reliable standard
errors for the parameter estimates, without the need of implementing a bootstrap
procedure. 

\section*{Appendix}
\begin{table}[ht]
\centering
\tiny{
\begin{tabular}{ccl}
\toprule
\bf $j$ & \bf $\#$ cat.  & \bf item description                           \\
\midrule
\multicolumn3c{Section CC} \\
\midrule
01    &  2  &  Short-term memory   (0 = ``recalls what recently happened (5 minutes)'', 1 = ``does not recall'')\\
02    &  2  &  Long-term memory (0 = ``keeps some past memories green'', 1 = ``does not keep some past memories green'') \\
03    &  2  &  Memory status (0 = ``recalls the actual season'', 1 = ``does not recall the actual season'')                   \\
04    &  2  &  Memory status (0 = ``recalls where is his room'', 1 = ``does not recall where is his room'')                   \\
05    &  2  &  Memory status  (0 = ``recalls the names and faces of the staff'', 1 = ``does not recall the names and faces of the staff'')  \\
06   &  2  &  Memory status (0 = ``recalls where he is'',  1 =``does not recall where he is'')                     \\
07   &  4  &  Decision about his daily activities (from 0 = ``independent decisions'' to 3 = ``unable to decide'') \\
08   &  3  &  Easily sidetracked  (from 0 = ``problems absent'' to 2 = ``problems worsened in the last week'')     \\
09   &  3  &  Altered perception or awareness of surrounding (from 0 = ``problems absent'' to 2 = ``problems worsened in the last week'') \\
10   &  3  &  Disorganized speech (from 0 = ``problems absent'' to 2 = ``problems worsened in the last week'')       \\
11   &  3  &  Restlessness movements (from 0 = ``problems absent'' to 2 = ``problems worsened in the last week'')     \\
12   &  3  &  Lethargic spans (from 0 = ``problems absent'' to 2 = ``problems worsened in the last week'')       \\
13  &  3  &  Change in the cognitive conditions during the day (from 0 = ``problems absent'' to 2 = ``problems worsened in the last week'')    \\
\midrule \multicolumn3c{ Section AVF} \\
 \midrule
14 &  4  &  Hearing (from 0 =``no hearing impairment'' to 3 = ``severe hearing impairment'')    \\
15   &  4  &  Ability to make itself understood  (from 0 = ``understood'' to 3 = ``seldom/never understood'')    \\
16   &  3  &  Clear language (from 0 = ``clear language'' to 2 = ``no language'')        \\
17  &  4  &  Ability to understand others (from 0 = ``understands'' to 3 = ``seldom/never understands'')    \\
18  &  5  &  Sight in conditions of adequate lighting (from 0 = ``no sight impairment'' to 4 = ``severe sight impairment'') \\
\midrule
\multicolumn3c{ Section HDB} \\
\midrule
19  &  3  &  Negative statements (from 0 = ``symptom not showed'' to 2 = ``symptom daily showed'')       \\
20  &  3  &   Repetitive questions (from 0 = ``symptom not showed'' to 2 = ``symptom daily showed'')       \\
21    &  3  &   Repetitive verbalizations (from 0 = ``symptom not showed'' to 2 = ``symptom daily showed'')       \\
22   &  3  &   Persistent anger with himself or others (from 0 = ``symptom not showed'' to 2 = ``symptom daily showed'')      \\
23     &  3  &   Self deprecation disesteem (from 0 = ``symptom not showed'' to 2 = ``symptom daily showed'')      \\
24   &  3  &   Fears that are not real (from 0 = ``symptom not showed'' to 2 = ``symptom daily showed'')       \\
25    &  3  &   To believe himself to be dying (from 0 = ``symptom not showed'' to 2 = ``symptom daily showed'')     \\
26     &  3  &   To complain about his health (from 0 = ``symptom not showed'' to 2 = ``symptom daily showed'')     \\
27     &  3  &   Repeated events anxiety (from 0 = ``symptom not showed'' to 2 = ``symptom daily showed'')\\
28    &   3  &   Unpleasant mood in morning (from 0 = ``symptom not showed'' to 2 = ``symptom daily showed'')    \\
29   &   3  &   Insomnia/problems with sleep (from 0 = ``symptom not showed'' to 2 = ``symptom daily showed'')    \\
30 &   3  &   Expressions of sad-faced (from 0 = ``symptom not showed'' to 2 = ``symptom daily showed'')       \\
31  &    3  &   Easily tears  (from 0 = ``symptom not showed'' to 2 = ``symptom daily showed'')      \\
32  &    3  &   Repetitive movements (from 0 = ``symptom not showed'' to 2 = ``symptom daily showed'')     \\
33  &   3  &   Abstention from activities of interest (from 0 = ``symptom not showed'' to 2 = ``symptom daily showed'')       \\
34  &     3  &   Reduced local interactions (from 0 = ``symptom not showed'' to 2 = ``symptom daily showed'')       \\
35  &    4  &   To wander aimlessly (from 0 = ``problem absent'' to 3 = ``problem daily encountered'')      \\
36  &   4  &    Offensive language (from 0 = ``problem absent'' to 3 = ``problem daily encountered'')    \\
37  &     4  &    Physically aggressive (from 0 = ``problem absent'' to 3 = ``problem daily encountered'')       \\
38  &    4  &    Socially inappropriate behavior (from 0 = ``problem absent'' to 3 = ``problem daily encountered'')      \\
39  &      4  &    To refuse assistance (from 0 = ``problem absent'' to 3 = ``problem daily encountered'')     \\
\midrule \multicolumn3c{ Section ADL} \\
\midrule
40  &  5  &   Moving to/from lying position (from 0 = ``independent'' to 4 = ``totally dependent'')    \\
41  &   5  &   Moving to/from bed, chair, wheelchair (from 0 = ``independent'' to 4 = ``totally dependent'')     \\
42  &  5  &   Walking between different points within the room (from 0 = ``independent'' to 4 = ``totally dependent'')     \\
43  &   5  &   Walking in the corridor (from 0 = ``independent'' to 4 = ``totally dependent'')       \\
44  & 5  &   Walking into the nursing home ward (from 0 = ``independent'' to 4 = ``totally dependent'')       \\
45  &  5  &   Walking outside the nursing home ward (from 0 = ``independent'' to 4 = ``totally dependent'')    \\
46  &   5  &   Dressing (from 0 = ``independent'' to 4 = ``totally dependent'')     \\
47  &  5  &   Eating  (from 0 = ``independent'' to 4 = ``totally dependent'')      \\
48  &  5  &   Using the toilet room (from 0 = ``independent'' to 4 = ``totally dependent'')       \\
49  &  5  &   Personal hygiene  (from 0 = ``independent'' to 4 = ``totally dependent'')    \\
50  &   5  &   Taking full-body bath/shower (from 0 = ``independent'' to 4 = ``totally dependent'')     \\
51  &  4  &   Balance problems (from 0 = ``does not have balance problems'' to 3 = ``needs physical assistance'')       \\
52  &   3  &   Mobility in the neck (0 = ``no limitation'', 1 = ``unilateral limitation'', 2 = ``bilateral limitation'')      \\
53  &    3  &   Mobility in the arm including shoulder or elbow (0 = ``no limitation'', 1 = ``unilateral limitation'', 2 = ``bilateral limitation'')        \\
54  &   3  &   Movements of the hand including wrist or finger     (0 = ``no limitation'', 1 = ``unilateral limitation'', 2 = ``bilateral limitation'')    \\
55  &   3  &   Mobility in the leg and hip (0 = ``no limitation'', 1 = ``unilateral limitation'', 2 = ``bilateral limitation'')         \\
56  &      3  &   Mobility in the foot and ankle (0 = ``no limitation'', 1 = ``unilateral limitation'', 2 = ``bilateral limitation'')         \\
57  &     3  &   Other movements (0 = ``no limitation'', 1 = ``unilateral limitation'', 2 = ``bilateral limitation'')          \\
\midrule \multicolumn3c{ Section I}\\
\midrule
58   &  5  &     Fecal incontinence (from 0 = ``continence'' to 4 = ``incontinence'')      \\
59       &  5  &     Urinary incontinence (from 0 = ``continence'' to 4 = ``incontinence'')        \\
60     &  2  &     Elimination of feces (0 = ``adequate'', 1 = ``not adequate'')      \\
\midrule \multicolumn3c{ Section NF}\\
\midrule
61  &   2  &      Chewing problem (0 = ``no problem'', 1 = ``problems'')        \\
62  &  2  &      Swallowing problem  (0 = ``no problem'', 1 = ``problems'')         \\
63  &  2  &      Mouth pain (0 = ``no problem'', 1 = ``problems'')           \\
64  &  2  &      Taste of many foods (0 = ``does not complain'', 1 = ``complains'')       \\
65  &  2  &      Hungry  (0 = ``does not complain'', 1 = ``complains'')   \\
66  & 2  &      Food on his plate (0 = ``does not leave it'', 1 = ``leaves it'')      \\
\midrule \multicolumn3c{ Section DD}\\
\midrule
67  &   2  &      Debris present in mouth prior to going to bed at night (0 = ``problem absent'', 1 = ``problem present'') \\
68  &   2  &      Dentures/removable bridge (0 = ``absent'', 1 = ``present'') \\
69  &  2  &     Some/all natural teeth lost and does not have/does not use dentures (or partial plates) (0 = ``problem absent'', 1 = ``problem present'') \\
70  &   2  &     Broken, loose, or carious teeth  (0 = ``problem absent'' 1 = ``problem present'')         \\
71  &   2  &     Inflamed gums, swollen or bleeding gums, oral abscesses, ulcers or rashes  (0 = ``problem absent'', 1 = ``problem present'') \\
72  &  2  &     Dentures or removable bridge daily cleaned by resident or staff (0 = ``absent'', 1 = ``present'') \\
     \midrule \multicolumn3c{ Section SC}\\
\midrule
73  &   5 &      Pressure ulcer (from 0 = ``no pressure ulcer'' to 4 = ``stage 4'')   \\
74  &   5 &      Stasis ulcers  (from 0 = ``no pressure ulcer'' to 4 = ``stage 4'')    \\
75  &    2 &      Resolved or cured ulcer (0 =``absent'', 1 = ``present'')    \\
   \bottomrule
\end{tabular}}\vspace*{0.2cm}
\small{\caption{\em Description of the full set of items.}}\label{tab:app}
\end{table}

\bibliographystyle{apalike}
\bibliography{reference}
\end{document}